\documentclass[11pt]{llncs}
\usepackage{graphicx}
\usepackage{graphics}
\usepackage{epsfig}
\usepackage{epstopdf}
\usepackage{amsmath}
\usepackage{latexsym}
\usepackage{wrapfig}
\usepackage{amsfonts}
\usepackage{cite}

\makeatletter

\newcommand{\Rmnum}[1]{\expandafter\@slowromancap\romannumeral #1@}
\makeatother

\newcommand{\deleted}[1]{}
\begin{document}

\title{Linear Time Algorithm for Projective Clustering}
\author{Hu Ding   \hspace{0.25in} Jinhui Xu}
\institute{
 Department of Computer Science and Engineering\\
State University of New York at Buffalo\\
 \email{\tt \{huding, jinhui\}@buffalo.edu}\\
}

\maketitle

\thispagestyle{empty}

\begin{abstract}
Projective clustering is a problem with both theoretical and practical importance and has received a great deal of attentions in recent years. Given a set of points $P$ in $\mathbb{R}^{d}$ space, projective clustering is to find a set $\mathbb{F}$ of $k$ lower dimensional $j$-flats so that the average distance (or squared distance) from points in $P$ to their closest flats is minimized. 
Existing approaches for this problem are mainly based on adaptive/volume sampling or core-sets techniques which suffer from several  limitations.   
In this paper, we present  the first uniform random sampling based approach for this challenging problem and achieve linear time solutions for three cases, general projective clustering, regular projective clustering, and $L_{\tau}$ sense projective clustering. For the general projective clustering problem, we show that for any given small numbers $0<\gamma, \epsilon <1$, our approach first removes $\gamma|P|$ points as outliers and then determines $k$ $j$-flats  to cluster  the remaining points into $k$ clusters with an objective value no more than $(1+\epsilon)$ times of the optimal for all points. 
For regular projective clustering, we demonstrate that when the input points satisfy some reasonable assumption on its input, our approach for the general case can be extended to yield a PTAS for all points. For $L_{\tau}$ sense projective clustering, we show that our techniques for both the general and regular cases can be naturally extended to the $L_{\tau}$ sense projective clustering problem for any $1 \le \tau < \infty$. Our results are based on several novel techniques, such as slab partition, $\Delta$-rotation, symmetric sampling, and recursive projection, and can be easily implemented for applications.

\end{abstract}

\newpage

\pagestyle{plain}
\pagenumbering{arabic}
\setcounter{page}{1}
\vspace{-0.2in}
\section{Introduction}
\vspace{-0.1in}

Projective clustering for a set $P$ of $n$ points in $\mathbb{R}^{d}$ space is to find a set $\mathbb{F}$ of $k$ lower dimensional $j$-flats so that the average distance (by certain distance measure) from points in $P$ to their closest flats is minimized. Depending on the choices of $j$ and $k$, the problem has quite a few different variants. For instance, when $k=1$, the problem is to find a $j$-flat to fit a set of points and is often called shape fitting problem. On the contrary, when $j=1$, the problem is to find $k$ lines to cluster a point set, and thus is called $k$-line clustering. In this paper, we mainly consider the $L_{2}$ sense projective clustering, i.e., minimizing the average squared distances to the resulting flats. We also consider extensions to regular projective clustering and $L_{\tau}$ sense projective clustering for any integer $1 \le \tau < \infty$, where the regular projective clustering is for  points whose projection on its optimal fitting flat have bounded coefficient of variation along any direction.



\noindent\textbf{Previous results: }Projective clustering is related to many theoretical problems such as shape fitting, matrix approximation, etc., as well as  numerous applications in applied domains. Due to its importance in both theory and applications, in recent years, a great deal of effort has devoted to solving this challenging problem and a number of promising techniques have been developed \cite{AHV04,FFS,FL11,HW04,HV02,HV03,AM04,AWY,AY00,PJA,HV02,SV07,DRV,EV05,VX12}. From methodology point of view, Agarwal {\em et al.} \cite{AHV04} first introduced a structure called {\em kernel set} for capturing the extent of a point set and used it to derive a number of algorithms related to the projective clustering problem. Har-Peled {\em et al.} \cite{HW04,HV03} presented algorithms for shape fitting problem based on kernel set and core-sets. The core-set concept has also been extended to more general projective clustering problems \cite{HV02,VX12,EV05}, and has proved to be effective for many other problems  \cite{AHV05,BC03,BHI02,Cl08,KMY02,HRZ}. 
 Another main approach for projective clustering is dimension reduction through adaptive sampling \cite{SV07,DRV}. From time efficiency point of view, most of the existing algorithms for projective clustering problems have super-linear dependency on the size $n$ of the point set. Several linear or near linear time (on $n$) algorithms were also previously presented. In \cite{AMV}, Agarwal {\em et al.} presented a near linear time algorithm for $k$-line clustering with $L_{\infty}$ sense objective. In \cite{EV05}, Edwards and Varadarajan introduced a near linear  time algorithm for integer points  and with $L_{\infty}$ sense objective. In \cite{VX12}, Varadarajan and Xiao designed a near linear time algorithm for $k$-line clustering and general projective clustering on integer points with $L_{1}$ sense objective. Furthermore, \cite{FFS,FL11} present a linear time bicriteria approximation algorithm with $L_1$, $L_2$ and $L_\infty$ sense.
 
\noindent\textbf{Relations with subspace approximation:} A problem closely related to $j$-flat fitting is the low rank matrix approximation problem whose objective is to find a lower dimensional subspace, rather than a flat, to approximate the original matrix (which is basically a set of column points). For this problem,  Frieze {\em et al.} introduced an elegant method based on random sampling \cite{FKV}. Their method additively approximates the original matrix, but unfortunately is not exact PTAS. To achieve a PTAS,    Deshpande {\em et al.} presented a volume sampling based approach to generate $j$-subspaces \cite{DRV}.  Their algorithm works well for the single $j$-flat/subspace fitting problem, and can also be extended to projective clustering problem (but with relatively high time complexity). Shyamalkumar {\em et al.} present an algorithm for subspace approximation with any $L_\tau$ sense objective, for $\tau\geq 1$ \cite{SV07}.

\vspace{-0.15in}
\section{Main Results and Techniques}
\label{sec-main}
\vspace{-0.12in}


\begin{definition}[$L_{\tau}$ Sense $(k,j)$-Projective Clustering and $j$-Flat Fitting]
\label{def-pc}
Given a point set $P$ in $\mathbb{R}^{d}$ space, and three integers $k\geq 1$, $1\leq j\leq d$ and $1\leq \tau<\infty$,  an $L_{\tau}$ sense  $(k,j)$-projective clustering is to find $k$ $j$-dimensional flats $\mathbb{F}=\{\mathcal{F}_{1}, \cdots, \mathcal{F}_{k}\}$ in $\mathbb{R}^{d}$ space such that $\frac{1}{|P|}\sum_{p\in P}\min_{1\leq i\leq k}||p,\mathcal{F}_{i}||^\tau$ is minimized. When $k=1$, it is a $j$-flat fitting problem.
\end{definition}
\vspace{-0.08in}

In this paper, we assume both $k$ and $j$ are constant. $||p,\mathcal{F}||$ is the closest distance from $p$ to $\mathcal{F}$.

\vspace{-0.17in}
\subsection{Main Results}
\vspace{-0.08in}

In this paper, we mainly focus on the case of $\tau=2$ on arbitrary points (i.e., general projective clustering), and then extend the ideas to two other cases, regular projective clustering and $L_{\tau}$ sense projective clustering for any integer $1 \le \tau < \infty$. We present a uniform approach, purely based on random sampling, to achieve linear time solutions for all three cases. 

\begin{itemize} 
\item {\bf General $(k,j)$-projective clustering:} For arbitrary point set $P$ and small constant numbers $0 < \gamma, \epsilon <1$, our approach leaves out a small portion (i.e., $\gamma |P|$) of the input points as outliers, and finds, in  $O(2^{poly(\frac{kj}{\epsilon\gamma})}nd)$ time, $k$ $j$-flats to cluster the remaining points so that their objective value is no more than $(1+\epsilon)$ times of the optimal value on the whole set $P$.   Our result relies on  several novel techniques, such as symmetric sampling, slab partition, $\Delta$-rotation, and recursive projection. 

\item {\bf Regular projective clustering:} When the input point set $P$ has  regular distribution on its clusters, our approach yields a PTAS solution for the whole point set $P$ in the same time bound. The regularity of $P$  is measured based on the Coefficient of Variation (CV) on the projection of its points along any direction on their optimal fitting flat. $P$ is regular if CV has a bounded value. Since many commonly encountered distributions, which are often used to model various data or noises in experiments, 
are regular (such as Gaussian distribution, Erlang distribution, etc), our result,  thus, has a wide range of  potential applications.    

 
\item {\bf $L_{\tau}$ sense projective clustering:} Our approach can also be extended to $L_{\tau}$ sense projective clustering for any $1 \leq \tau<\infty $ and with the same time bound. We show that each technique used for the general and regular projective clustering (i.e., the case of $\tau=2$) can be extended to achieve similar results.  

\end{itemize}

\noindent\textbf{Comparsons with previous results:} As mentioned earlier,  existing works on projective clustering can be  classified into two categories: (a) adaptive sampling (or volume sampling) based approaches \cite{DRV,SV07} and (b) Core-sets based approaches \cite{EV05,VX12}. Often, (a) can efficiently solve the single flat fitting problem  (i.e., subspace approximation), 
but its extension to projective clustering requires a running time (i.e., $O(d(n/\epsilon)^{jk^3/\epsilon})$) much higher than the desired (near) linear time. (b) can solve projective clustering in near linear time, but the input must be integer points and within a polynomial range (i.e., $(mn)^{10}$) in any coordinate. The main advantages of our approach are:  (1) its  linear time complexity, (2)  do not need to have any assumption on its input (if a small fraction of outliers is allowed), (3) achieve linear time PTAS for regular points, (4) simple and can be easily implemented for applications.


\vspace{-0.15in}
\subsection{Key Techniques}
\vspace{-0.05in}

Our approach is based on a key result  in \cite{IKI}, which estimates the mean point of large point set by a small random sample whose size is independent of the size and dimensionality  of the original set.  
This result is widely used in many areas, especially in $k$-means clustering \cite{KSS04,KSS05,KSS10}. Since projective clustering is a generalization of $k$-means clustering, where the mean point is simply a $0$-dimensional flat, it is desirable  to generalize this uniform random sampling technique to the more general flat fitting and projective clustering problems (without relying on adaptive or volume sampling or core-sets techniques). 


To address this issue, we show that 
after taking a random sample $\mathcal{S}$, it is impossible to generate a proper fitting flat if we simply compute the mean of $\mathcal{S}$ as in \cite{IKI}. Our 
 key idea is to use {\em Symmetric Sampling} technique to consider not only $\mathcal{S}$, but also $-\mathcal{S}$, which is the symmetric point set of $\mathcal{S}$ with respect to the mean point $o$ of the input set $P$. Intuitively, if we enumerate the mean point of every subset of $\mathcal{S}\cup -\mathcal{S}$, there must exist one such point  $p$ that not only locates close to the optimal fitting flat,  but also is far away from $o$. This means that $p$ can define one dimension of the fitting flat, and thus we can reduce the $j$-flat fitting problem to a $(j-1)$-flat fitting problem by projecting all points to some $(d-1)$ dimensional subspace. If recursively use the strategy $j$ times, which is called {\em \textbf{Recursive Projection}}, we can get one proper flat. With this flat fitting technique, we can naturally extend it to projective clustering.
 


\vspace{-0.1in}

\section{Hyperbox Lemma and Slab Partition}
\label{sec-cuboid}
\vspace{-0.1in}

In this section, we present two standalone results, {\em Hyperbox Lemma} and {\em Slab Partition}, which are used for proving our key theorem (i.e., Theorem \ref{the-alg1}) in Section \ref{sec-symmetric}. 


\begin{definition}[Slab and Amplification]
\label{def-sa}
Let $o$ and $s$ be two points in $\mathbb{R}^d$, and $\Omega$ and $-\Omega$ be the two hyperplanes perpendicular to vector $\overrightarrow{os}$ and passing through $s$ and $-s$ respectively, where $-s$ is $s$'s symmetric point about $o$.  
The region bounded by $\Omega$ and $-\Omega$ is called the Slab  determined by $\overrightarrow{os}$ (denoted as $\mathcal{R}$). Further, let $s'$ be a point collinear with $o$ and $s$ with $\frac{||\overrightarrow{os}'||}{||\overrightarrow{os}||}=\lambda$.  Then the Slab $\mathcal{R}'$ determined by $\overrightarrow{os}'$ is called an amplification of $\mathcal{R}$ by a factor $\lambda$ (see Figure \ref{fig-defslab}).
\end{definition}

\vspace{-0.2in}
\subsection{Hyperbox Lemma}
\label{sec-cuboidlemma}
\vspace{-0.1in}
\vspace{-0.25in}
\begin{figure}[]
\vspace{-0.15in}
\begin{minipage}[t]{0.5\linewidth}
  \centering
  \includegraphics[height=1.2in]{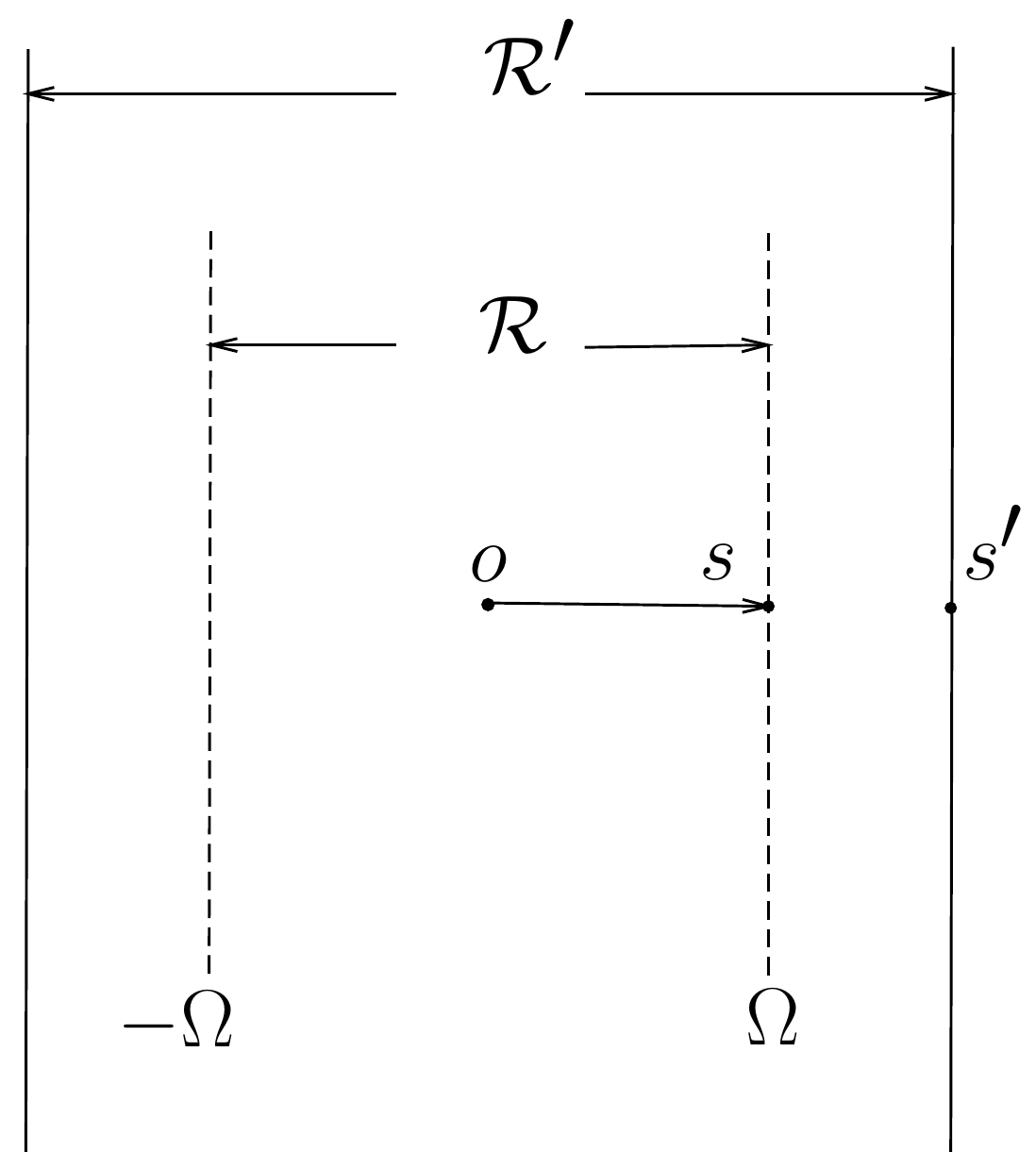}
  \vspace{-0.1in}
     \caption{An example illustrating Definition \ref{def-sa}.}
  \label{fig-defslab}
\end{minipage}
\begin{minipage}[t]{0.5\linewidth}
\centering
  \includegraphics[height=1.4in]{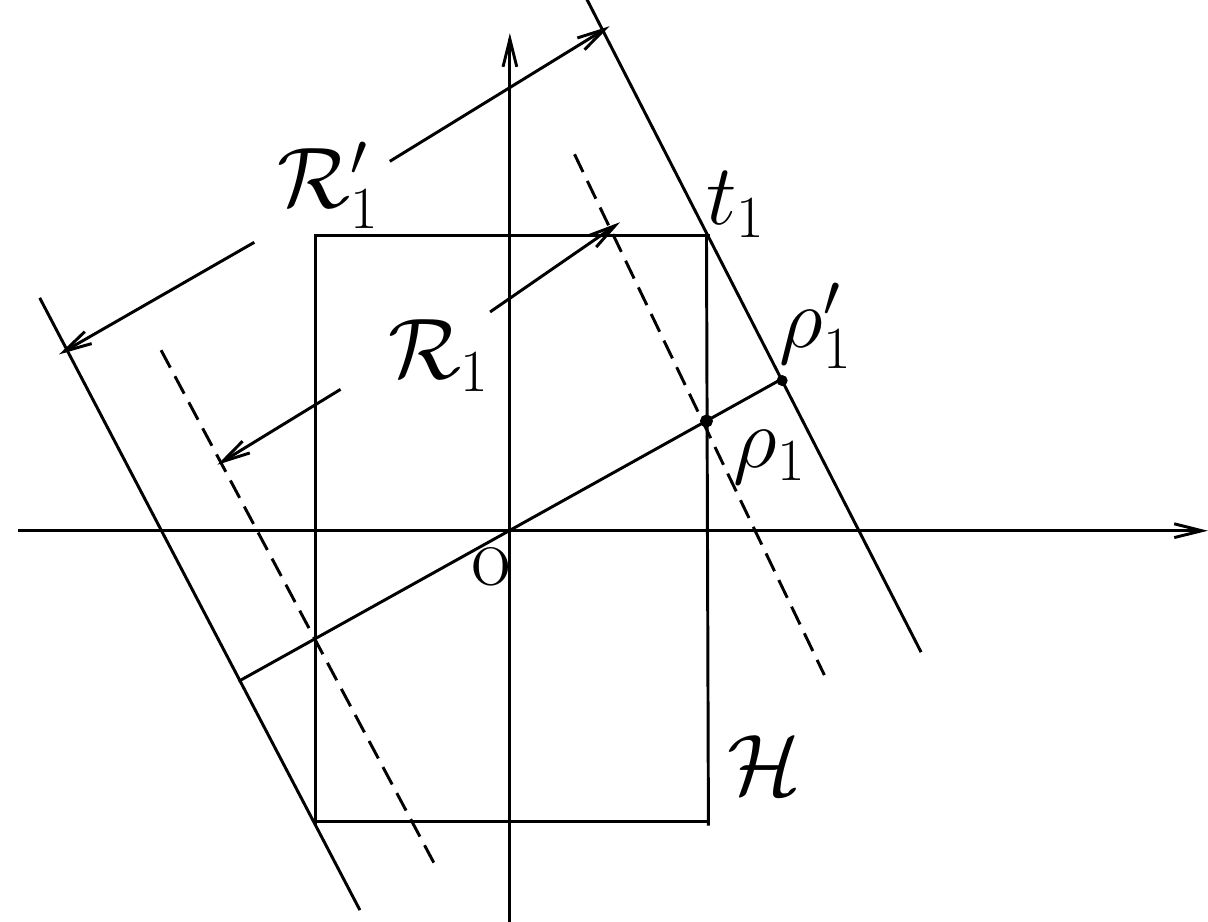}
  \vspace{-0.1in}
       \caption{An example illustrating Lemma \ref{lem-cuboid}.}
  \label{fig-cuboid}
  \end{minipage}
  \vspace{-0.2in}
\end{figure}
\vspace{-0.2in}
\begin{lemma}[Hyperbox Lemma]
\label{lem-cuboid}
Let $\mathcal{H}$ be a hyperbox in $\mathbb{R}^j$, and $o$ be its center. Let $\{f_1, \cdots, f_j\}$ be $j$ facets (i.e., $(j-1)$-dimensional faces) of $\mathcal{H}$ with different normal directions (i.e., no pair are parallel to each other), and $\rho=\{\rho_1, \cdots, \rho_j\}$ be $j$ points with each $\rho_{i}$, $1 \le i \le j$, incident to $f_{i}$. Then there exists one point $\rho_{l_0} \in \rho$  such that the slab determined by $\overrightarrow{o\rho}_{l_0}$ contains $\mathcal{H}$ after amplifying by a factor no more than $\sqrt{j}$.
%
%
\end{lemma}
\begin{proof}

Let $a_1, \cdots, a_j$ be the $j$ side lengths of $\mathcal{H}$.  For each $1\le l \le j$, denote the slab  determined by $\overrightarrow{o\rho}_l$ as $\mathcal{R}_l$ (with two bounding hyperplanes $\Omega_{l}$ and $-\Omega_{l}$), and its minimal amplification, which is barely enough to contain $\mathcal{H}$, as   $\mathcal{R}'_l$ (i.e., its two bounding hyperplanes $\Omega_{l}'$ and $-\Omega_{l}'$ support  $\mathcal{H}$).
%
%
Let $t_{l}$ be a point in $\Omega_{l}' \cap \mathcal{H}$ (i.e., a point on the (possibly $0$-dimensional) touching face of $\Omega'_{l}$ and $\mathcal{H}$), and $\rho'_{l}$ be the intersection point of $\Omega'_{l}$ and the supporting line of $o$ and $\rho_{l}$ (see Figure \ref{fig-cuboid}).     
Then we have $||o-\rho'_l||\leq ||o-t_l ||\leq \sqrt{\sum^j_{w=1} a^2_w}$, and $||o- \rho_l ||\geq a_l$ . Thus, we know that the amplification factor $\frac{||o-\rho'_l ||}{||o-\rho_l ||}\leq\frac{\sqrt{\sum^j_{w=1} a^2_w}}{a_l}$. Let $a_{l_0}=\max\{a_1, \cdots, a_j\}$, then we have $\frac{||o-\rho'_{l_0} ||}{||o-\rho_{l_0} ||}\leq\frac{\sqrt{\sum^j_{w=1} a^2_w}}{a_{l_0}}\leq\frac{\sqrt{ja^2_{l_0}}}{a_{l_0}}= \sqrt{j}$. Thus the lemma is true.
\qed
\end{proof}

\vspace{-0.23in}
\subsection{Slab Partition}
\label{sec-slab}
\vspace{-0.07in}
\begin{definition}[Slab Partition]
\label{def-slab}
Let $o$ be the origin of $\mathbb{R}^j$, and $\overrightarrow{ou}_1, \cdots, \overrightarrow{ou}_j$ be the $j$ orthogonal vectors defining the coordinate system of $\mathbb{R}^j$. The following partition is called Slab Partition on $\mathbb{R}^j$: $\Pi_l=\Pi_{l-1}\cap \mathcal{R}_l$ for $1\leq l\leq j$, where $\Pi_0=\mathbb{R}^j$, $\mathcal{R}_l$ is the Slab determined by $\overrightarrow{o u}'_l$, and $u'_l$ is some point on the ray of $\overrightarrow{o u}_l$ (see Figure \ref{fig-slab}). 
\end{definition}
 \vspace{-0.25in}



%
%
%
\begin{figure}[]
\vspace{-0.15in}
\begin{minipage}[t]{0.4\linewidth}
  \centering
  \includegraphics[height=1.3in]{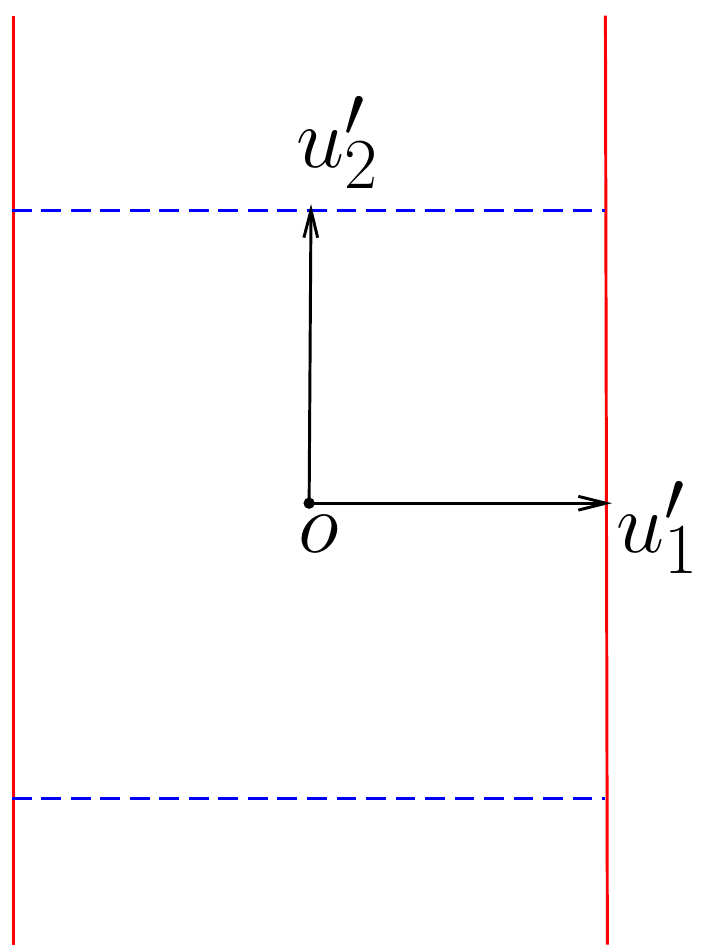}
  \vspace{-0.1in}
     \caption{An example of 2D slab partition.}
  \label{fig-slab}
\end{minipage}
\begin{minipage}[t]{0.6\linewidth}
\centering
  \includegraphics[height=1.3in]{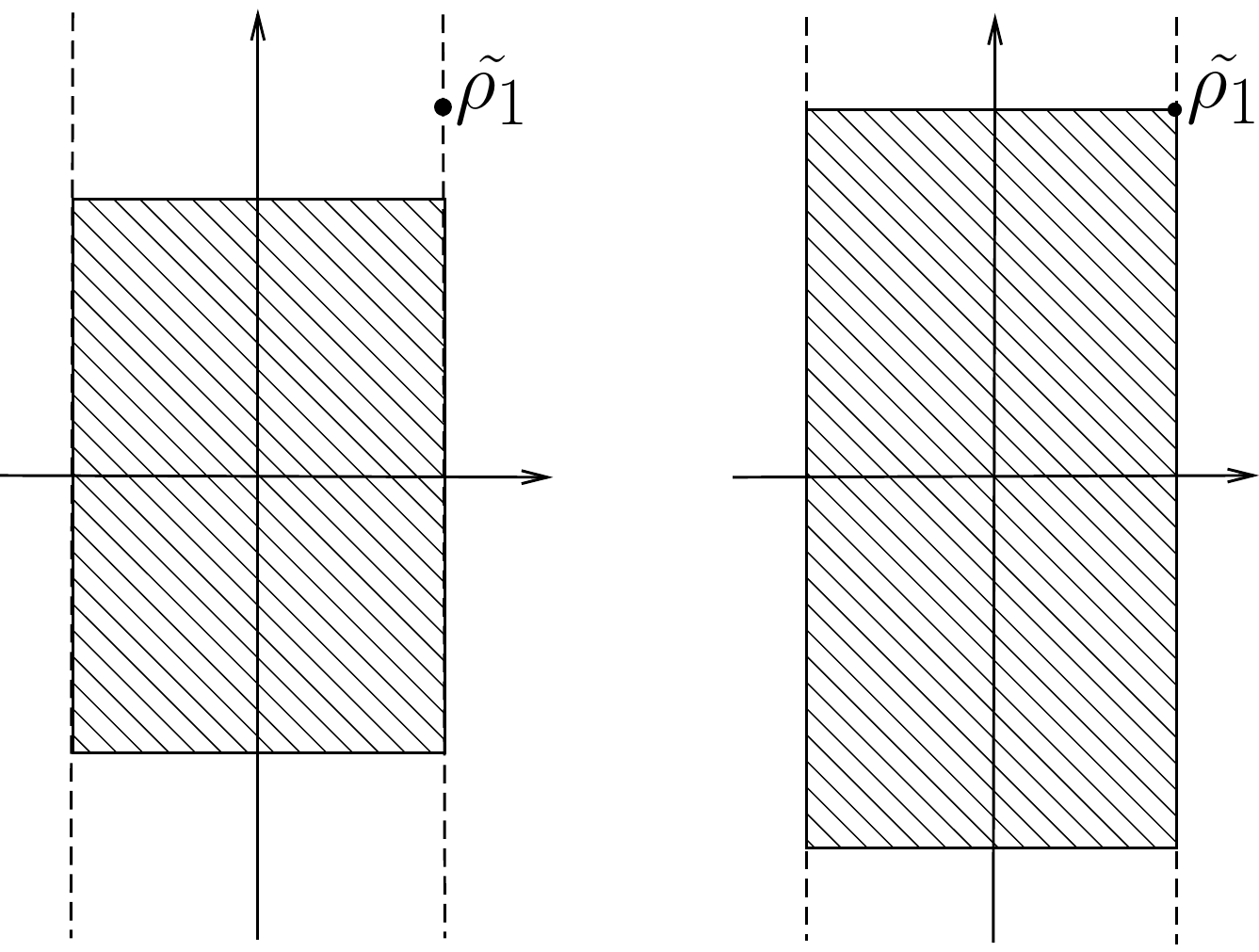}
  \vspace{-0.1in}
       \caption{The right cuboid is enlarged from the left one respect to $\rho_1$.}
  \label{fig-enlarge}
  \end{minipage}
  \vspace{-0.2in}
\end{figure}
 \vspace{-0.2in}


\begin{lemma}
\label{lem-slab}
Let $\{\Pi_0, \Pi_1, \cdots, \Pi_j\}$ be a slab partition in $\mathbb{R}^j$, and $\{\mathcal{R}_1, \cdots, \mathcal{R}_j\}$ be the corresponding partitioning slabs. Let $\{\rho_1, \cdots, \rho_j\}$ be the $j$ points such that $\rho_l\in\Pi_{l-1}\cap\partial\mathcal{R}_l$ for $1\leq l\leq j$, where $\partial\mathcal{R}_l$ is the  bounding hyperplane of $\mathcal{R}_l$. Then there exists a point $\rho_{l_0}$, such that the slab determined by $\overrightarrow{o\rho}_{l_0}$ contains  $\Pi_j$ after amplified by a factor of $\sqrt{j}$.
%
%
%
%
%
%
%
\end{lemma}
\begin{proof}
It is easy to see that $\Pi_j=\cap^j_{l=1}\mathcal{R}_l$, which is a hyperbox in $\mathbb{R}^j$. Thus, it is natural to use Lemma \ref{lem-cuboid} to prove the lemma. For this purpose, we let 
 $\mathcal{S}_l$, $1\le l \le j$, be one of the bounding hyperplanes of  $\mathcal{R}_l$ with $\rho_l$ incident to it. 
 For any $w<l$, from slab partition we know that the whole $\Pi_l$ locates inside $\mathcal{R}_w$. Thus, $\rho_l$ also locates inside $\mathcal{R}_w$. Let $f_l=\Pi_j\cap\mathcal{S}_l$. Thus the $j$ facets $\{f_1, \cdots, f_j\}$ of $\Pi_j$ point (i.e., their normal directions) to different directions.  
 Note that since $f_l$ is only a subregion of $\mathcal{S}_l$, $\rho_l$ is possibly outside of $f_l$. Thus we consider the following two cases,  (a) every $\rho_l$ locates inside $f_l$ for $1\leq l\leq j$ and (b) there exists some $\rho_l$ locates outside of $f_l$.

%

For case (a), the lemma follows from Lemma \ref{lem-cuboid} after replacing $\mathcal{H}$ by $\Pi_{j}$. For case (b), our idea is to reduce it to case (a) through the following procedure.
\begin{enumerate}
\item Initialize a set of points $\{\tilde{\rho}_1, \cdots, \tilde{\rho}_j\}$ with $\tilde{\rho}_l=\rho_l$ for $1\leq l\leq j$.
\item Set $l=1$. Do the following steps until $l>j$.
\begin{enumerate}
\item Set $w=l+1$. Do the following steps until $w>j$.
\begin{enumerate}

\item If $\tilde{\rho}_l$ is outside of $\Pi_w$, first amplify $\mathcal{R}_w$ until it touches $\tilde{\rho}_l$ (see Fig. \ref{fig-enlarge}), and then set $\tilde{\rho}_w=\tilde{\rho}_l$.
\item $w=w+1$. 
\end{enumerate}
\item $l=l+1$.
\end{enumerate}
\end{enumerate}
\begin{claim}
After the above procedure, $\{\tilde{\rho}_1, \cdots, \tilde{\rho}_j\}$ becomes a case (a)  set with respect to the  amplified $\Pi_j=\cap^j_{l=1}\mathcal{R}_l$.
\end{claim}
To show this claim, we observe that there are two loops in the procedure. In the first loop, each $l$-th round guarantees that $\tilde{\rho}_l$ locates inside (or on the boundary) of the $l$-th facet of the enlarged $\Pi_j$. Note that $\tilde{\rho}_l$ is always inside of $\mathcal{R}_w$ for $w\leq l$.  Thus the second loop only starts from $w=l+1$. After amplifying $\mathcal{R}_w$, the original $\tilde{\rho}_w$ will no longer  be  on $\partial\mathcal{R}_w$.  Thus replacing $\tilde{\rho}_w$ by $\tilde{\rho}_l$ will keep it on the boundary of  $\partial\mathcal{R}_w$. Thus, after finishing the two loops, $\{\tilde{\rho}_1, \cdots, \tilde{\rho}_j\}$ will become a case (a) set with respect to the new $\Pi_j$.

%
 
 \noindent\textbf{Note} that in case (b), the resulting $\{\tilde{\rho}_1, \cdots, \tilde{\rho}_j\}$ is actually a subset of the original  $\{\rho_1, \cdots, \rho_j\}$. Thus, we do not really need to perform the procedure to complete the reduction. 
 We only need to find the desired $\rho_{l_0}$ whose existence is ensured by Lemma \ref{lem-cuboid}. Thus, the lemma holds.
\qed
\end{proof}

\vspace{-0.3in}
\section{$\Delta$-Rotation and Symmetric Sampling}
\label{sec-technique}
 \vspace{-0.1in}
This section introduces several key techniques used in our algorithms. Let $\mathcal{F}$ be a $j$-dimensional flat and $P$ be a set of $\mathbb{R}^{d}$ points. We denote the average squared distance from $P$ to $\mathcal{F}$ as $\delta^{2}_{P,\mathcal{F}} =\frac{1}{|P|}\sum_{p\in P}$ $||p, \mathcal{F}||^2$, where $||p, \mathcal{F}||$ is the closest distance from $p$ to $\mathcal{F}$.  

\vspace{-0.15in}
\subsection{Flat Rotation and $\Delta$-Rotation}
\label{sec-rotate}
 \vspace{-0.12in}
In this section, we discuss flat rotation, and how it affects single flat fitting. 

\vspace{-0.05in}
\begin{figure}[]
\vspace{-0.15in}
\begin{minipage}[t]{0.55\linewidth}
  \centering
  \includegraphics[height=1in]{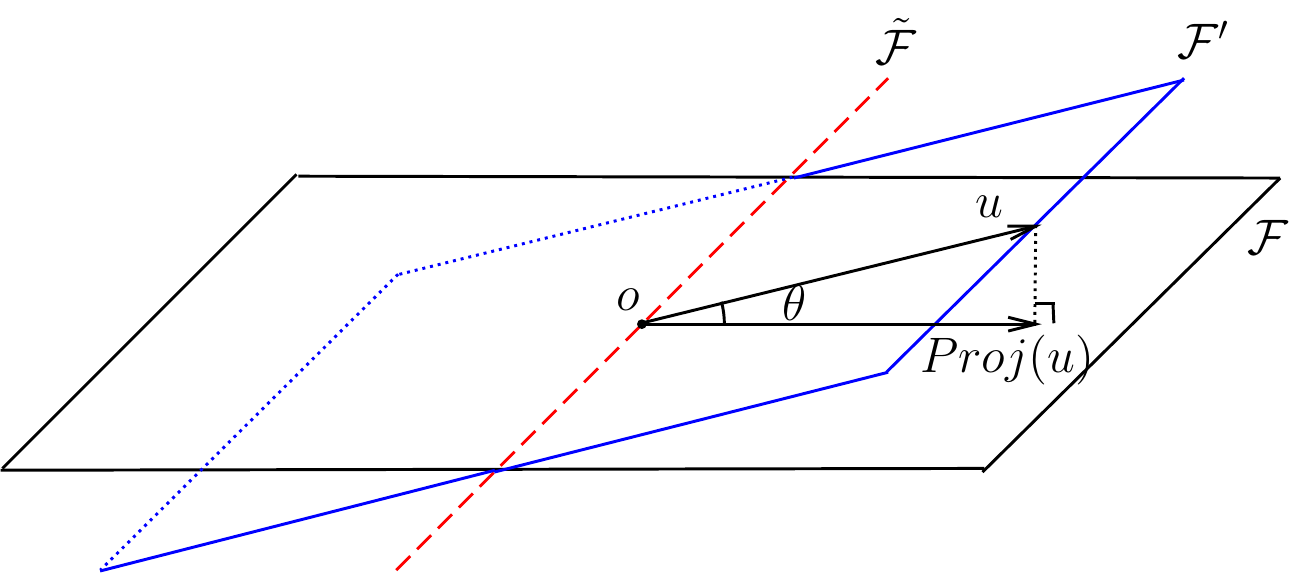}
  \vspace{-0.1in}
     \caption{An example illustrating Definition \ref{def-rotate}.}
  \label{fig-rotate}
\end{minipage}
\begin{minipage}[t]{0.45\linewidth}
\centering
  \includegraphics[height=0.9in]{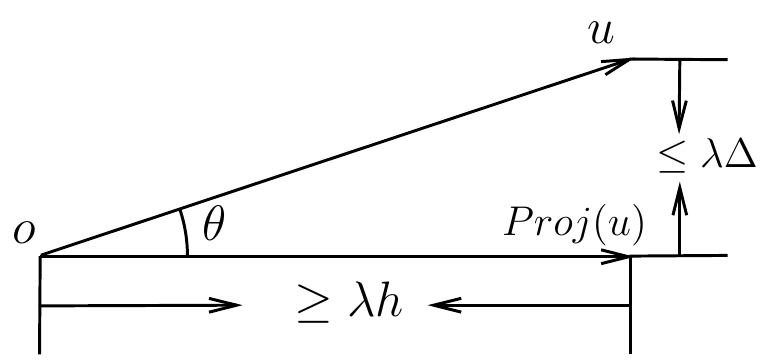}
  \vspace{-0.1in}
       \caption{An example illustrating Definition \ref{def-delta}.}
  \label{fig-delta}
  \end{minipage}
  \vspace{-0.1in}
\end{figure}
\vspace{-0.1in}
\begin{definition}[Flat Rotation]
\label{def-rotate}
Let $\mathcal{F}$ be a $j$-dimensional flat in $\mathbb{R}^d$,  $o$ be a point on $\mathcal{F}$, and $u$ be any given point in $\mathbb{R}^d$. Let $Proj(u)$ denote the orthogonal projection of $u$ on $\mathcal{F}$, and $\tilde{F}$ denote the $j-1$-dimensional face of $\mathcal{F}$ which is perpendicular to the vector $Proj(u)-o$. Then the flat $\mathcal{F}'$ spanned by $\tilde{F}$ and the vector $u-o$ is a rotation of $\mathcal{F}$ induced by the vector $u-o$, and the rotation angle $\theta$ is the angle between $u-o$ and $Proj(u)-o$ (see Fig. \ref{fig-rotate}). 
\end{definition}
 \vspace{-0.1in}

In the above definition, when there is no ambiguity about $o$, we also call the rotation is induced by $u$. 


%
%

\vspace{-0.1in}
\begin{definition}[$\Delta$-Rotation]
\label{def-delta}
Let $P$ be a point set and $\mathcal{F}$ be a $j$-dimensional flat in $\mathbb{R}^d$. Let $o$ be a point on $\mathcal{F}$,  $u$ be any given point in $\mathbb{R}^d$, and $h^2=\frac{1}{|P|}\sum_{p\in P}|<p-o, \frac{Proj(u)-o}{||Proj(u)-o||}>|^2$, where $Proj(u)$ is the orthogonal projection of $u$ on $\mathcal{F}$, and $<a,b>$ denotes the inner product of $a$ and $b$. Let $\mathcal{F}'$ be a rotation of $\mathcal{F}$ induced by the vector $u-o$ with angle $\theta$. Then it is a $\Delta$-rotation with respect to $P$ if $\theta \leq \arctan\frac{\Delta}{h}$.
\end{definition}
In the above definition, $h^{2}$ is the average squared projection length of each $p-o$  along the direction of $Proj(u)-o$.  Figure \ref{fig-delta} shows an example of $\Delta$-rotation. The following lemma shows  how the average squared distance $\delta^{2}_{P,\mathcal{F}}$ from $P$ to $\mathcal{F}$ (i.e.,  $\frac{1}{|P|}\sum_{p\in P}$ $||p, \mathcal{F}||^2$) changes after a $\Delta$-rotation.
\begin{lemma}
\label{lem-delta}
Let $P$ be a point set in $\mathbb{R}^{d}$, $\mathcal{F}$ be a $j$-dimensional flat, and $u$ be a point in $\mathbb{R}^d$.  If  $\mathcal{F}'$ is a $\Delta$-rotation (with respect to $P$) of $\mathcal{F}$ induced by the vector $u-o$ for some point $o\in \mathcal{F}$, then $\delta_{P,\mathcal{F}'} \le \delta_{P,\mathcal{F}} + \Delta$.
\end{lemma}
\begin{proof}
We use the same notations as in Definition \ref{def-delta}. For any $p\in P$, we let $u_p$ denote $|<p-o, \frac{Proj(u)-o}{||Proj(u)-o||}>|$,  and $Proj(p)$ denote its orthogonal projection on $\mathcal{F}$. Then by triangle inequality, we have
\small{
\begin{eqnarray}
||p, \mathcal{F}'||\leq ||p-Proj(p)||+||Proj(p), \mathcal{F}'||=||p, \mathcal{F}||+||Proj(p), \mathcal{F}'||. \label{for-1}
\end{eqnarray}
}
\normalsize
Meanwhile, since the rotation angle from $\mathcal{F}$ to $\mathcal{F}'$ is $\theta\leq\arctan\frac{\Delta}{h}$, we have $||Proj(p), \mathcal{F}'||=u_p \sin\theta\leq u_p \tan\theta\leq  \frac{\Delta}{h}u_p$. 
Plugging this into inequality (\ref{for-1}), we get
\small{
\begin{eqnarray}
||p, \mathcal{F}'||^2 &\leq& (||p, \mathcal{F}||+u_p \sin\theta)^2=||p, \mathcal{F}||^2+2 u_p \sin\theta||p, \mathcal{F}||+(u_p \sin\theta)^2 \nonumber\\
&\leq& ||p, \mathcal{F}||^2+2  \frac{\Delta}{h}u_p ||p, \mathcal{F}||+(  \frac{\Delta}{h}u_p)^2=||p, \mathcal{F}||^2+ \frac{\Delta}{\delta_{P,\mathcal{F}}}\cdotp 2\frac{\delta_{P,\mathcal{F}}}{h}u_p ||p, \mathcal{F}||+(  \frac{\Delta}{h}u_p)^2 \nonumber\\
&\leq& ||p, \mathcal{F}||^2+\frac{\Delta}{\delta_{P,\mathcal{F}}} (||p, \mathcal{F}||^2+(\frac{\delta_{P,\mathcal{F}}}{h}u_p)^2)+(\frac{\Delta}{\delta_{P,\mathcal{F}}})^2( \frac{\delta_{P,\mathcal{F}}}{h}u_p)^2 \nonumber\\
&=&(1+\frac{\Delta}{\delta_{P,\mathcal{F}}})||p, \mathcal{F}||^2+(\frac{\Delta}{\delta_{P,\mathcal{F}}}+(\frac{\Delta}{\delta_{P,\mathcal{F}}})^2)(\frac{\delta_{P,\mathcal{F}}}{h}u_p)^2, \label{for-2}
\end{eqnarray}
}
\normalsize
where the first inequality follows from $\sin\theta\leq\frac{\Delta}{h}$, and the second inequality follows from the fact that $2ab\leq a^2+b^2$ for any pair of real numbers $a$ and $b$. Summing both sides of (\ref{for-2}) over $p$,  we have
\small{
\begin{eqnarray}
\sum_{p\in P}||p, \mathcal{F}'||^2 &\leq& (1+\frac{\Delta}{\delta_{P,\mathcal{F}}})\sum_{p\in P}||p, \mathcal{F}||^2+(\frac{\Delta}{\delta_{P,\mathcal{F}}}+(\frac{\Delta}{\delta_{P,\mathcal{F}}})^2)\sum_{p\in P}(\frac{\delta_{P,\mathcal{F}}}{h}u_p)^2 \nonumber\\
&=&(1+\frac{\Delta}{\delta_{P,\mathcal{F}}})\sum_{p\in P}||p, \mathcal{F}||^2+(\frac{\Delta}{\delta_{P,\mathcal{F}}}+(\frac{\Delta}{\delta_{P,\mathcal{F}}})^2)\frac{\delta^2_{P,\mathcal{F}}}{h^2}\sum_{p\in P}(u_p)^2. \label{for-3}
\end{eqnarray}
}
\normalsize
Since $h^2=\frac{1}{|P|}\sum_{p\in P}(u_p)^2$, $\delta^2_{P,\mathcal{F}}=\frac{1}{|P|}\sum_{p\in P}||p, \mathcal{F}||^2$, and $\delta^2_{P,\mathcal{F}'}=\frac{1}{|P|}\sum_{p\in P}||p, \mathcal{F'}||^2$. (\ref{for-3}) becomes
\small{
\begin{eqnarray}
\delta^{2}_{P,\mathcal{F}'}&=&\frac{1}{|P|}\sum_{p\in P}||p, \mathcal{F}'||^2\leq(1+\frac{\Delta}{\delta_{P,\mathcal{F}}})\frac{1}{|P|}\sum_{p\in P}||p, \mathcal{F}||^2+(\frac{\Delta}{\delta_{P,\mathcal{F}}}+(\frac{\Delta}{\delta_{P,\mathcal{F}}})^2)\frac{1}{|P|}\sum_{p\in P}||p, \mathcal{F}||^2 \nonumber\\
&=&(1+\frac{\Delta}{\delta_{P,\mathcal{F}}})^2\frac{1}{|P|}\sum_{p\in P}||p, \mathcal{F}||^2=(\delta_{P,\mathcal{F}}+\Delta)^2. \label{for-4}
\end{eqnarray}
}
\normalsize
Thus, the lemma is true.
\qed
\end{proof}

\vspace{-0.25in}
\subsection{Symmetric Sampling}
\label{sec-symmetric}

\noindent\textbf{Algorithm Symmetric-Sampling}
\newline \textbf{Input:} A set $\mathcal{S}=\{s_{1}, s_{2}, \ldots, s_{m}\}$ of $\mathbb{R}^{d}$ points and a single point $o$ in $\mathbb{R}^{d}$.
\newline \textbf{Output:} A new point set $\overline{\mathcal{S}}$.
 \begin{enumerate}
  \item Initialize $\overline{\mathcal{S}}=\emptyset$.
 
  \item Construct a new point set $-\mathcal{S}=\{2o-s_{1}, \cdots, 2o-s_{m}\}$, which is the set of symmetric points of $\mathcal{S}$ (i.e., symmetric about $o$).

  \item For each subset of $\mathcal{S}\cup -\mathcal{S}$, add its mean point  into $\overline{\mathcal{S}}$.
  \end{enumerate}

Below is the main theorem about Algorithm Symmetric-Sampling.
\vspace{-0.1in}
\begin{theorem}
\label{the-alg1}
Let $P$ be a set of  $\mathbb{R}^d$ points,  and $\mathcal{S}$ be its random sample of size $r=\frac{4j^2}{\gamma\epsilon}\ln\frac{j}{\epsilon}$,  where $0<\gamma \le 1$ and $\epsilon >0$ are two small numbers. Let $\mathcal{F}$ be a $j$-dimensional flat,  $o$ be a given point on $\mathcal{F}$,  and  $\overline{\mathcal{S}}$ be the set of points returned by Algorithm Symmetric-Sampling on $\mathcal{S}$ and $o$.
Then with probability $(1-\epsilon)^4$,  $\overline{\mathcal{S}}$ contains one point $s$ such that the hyperplane $\mathcal{F}'$ rotated from $\mathcal{F}$ and induced by $\overrightarrow{os}$  satisfies the following inequality,
$|P'|\delta^{2}_{P',\mathcal{F}'} \le  (1+5\sqrt{j}r)^2 |P| \delta^{2}_{P,\mathcal{F}}$,
where $P'$ is a subset of $P$ with size $|P'|\geq (1-\frac{\gamma}{j})|P|$.
\end{theorem}
\vspace{-0.05in}
Before proving  Theorem \ref{the-alg1}, we first introduce the following two lemmas.
\vspace{-0.05in}
 \begin{lemma}
\label{lem-select}
Let $S$ be a set of $n$ elements, and $S'$ be a subset of $S$ with size
$|S'|=\alpha n$. If randomly select $\frac{t}{\ln(1+\alpha)}\ln\frac{t}{\eta}=O(\frac{t}{\alpha}\ln\frac{t}{\eta})$ elements from $S$, with probability at least $1-\eta$, the sample contains at least $t$ elements from $S'$ (see Appendix for proof).
\end{lemma}
\vspace{-0.05in}
The following lemma has been proved in \cite{IKI}. 
\vspace{-0.05in}
\begin{lemma}[\cite{IKI}]
\label{lem-dis}
Let $S$ be a set of $n$ points in $\mathbb{R}^d$ space,  and $T$ be a subset with cardinality $m$ randomly selected from $S$.  Let $\overline{x}(S)$ and $\overline{x}(T)$ be the mean points of $S$ and $T$ respectively. With probability $1-\eta$, $||\overline{x}(S)-\overline{x}(T)||^2<\frac{1}{\eta m}Var^{0}(S)$, where $Var^{0}(S)$ $=\frac{\sum_{s\in S}||s-\overline{x}(S)||^2}{n}$.
\end{lemma}


\begin{proof}[\textbf{of Theorem \ref{the-alg1}}]
 First, imagine that if we can show the existence of (1)  a subset $P'\subset P$ with size $|P'|\geq (1-\frac{\gamma}{j})|P|$ and  (2) a $j$-flat $\mathcal{F}'$ which is a $\Delta$-rotation (with $\Delta=\frac{5\sqrt{j}r}{\sqrt{1-\gamma/j}}\delta_{P,\mathcal{F}}$) of $\mathcal{F}$ with respect to $P'$, then by Lemma \ref{lem-delta}, we have $\delta_{P',\mathcal{F}'} \le \delta_{P',\mathcal{F}} + \Delta$.  
Meanwhile, since $\sum_{p\in P'}||p, \mathcal{F}||^2\leq \sum_{p\in P}||p, \mathcal{F}||^2$ and $|P'|\geq (1-\frac{\gamma}{j})|P|$, we have $\delta^{2}_{P',\mathcal{F}} \le \frac{1}{1-\gamma/j}\delta^{2}_{P,\mathcal{F}}$. Combining the two inequalities, we have  $ \delta_{P',\mathcal{F}'} \le (\sqrt{\frac{1}{1-\gamma/j}}+\frac{5\sqrt{j}r}{\sqrt{1-\gamma/j}}) \delta_{P,\mathcal{F}} \Rightarrow |P'|\delta^{2}_{P',\mathcal{F}'} \le  (1+5\sqrt{j}r)^2 |P| \delta^{2}_{P,\mathcal{F}}$.
This means that we only need to focus on proving the existence of such $P'$ and $\mathcal{F}'$.




Without loss of generality, we assume that the given point $o$ is the origin. To prove the theorem, we first assume that all points of $P$ locate on $\mathcal{F}$, which
would be the case if project all points of $P$ onto $\mathcal{F}$ or equivalently each point in this case can be viewed as the projection of some point of $P$ in $\mathbb{R}^{d}$. 



Now we consider the case that $P$ is in $\mathcal{F}$.  First, we  construct a slab partition  $\{\Pi_0, \Pi_1, \cdots, \Pi_j\}$ for the $j$-dimensional subspace spanned by 
$\mathcal{F}$ with $\{\mathcal{R}_1, \cdots, \mathcal{R}_j\}$ being the corresponding slabs such that $|P\cap\Pi_l|=(1-\frac{\gamma}{j^2})|P\cap\Pi_{l-1}|$. 
Clearly, this can be easily obtained by iteratively (starting from $l=1$) selecting the slab $R_{l}$ as the one which exclude the $\frac{\gamma}{j^2}|P \cap \Pi_{l-1}|$ points whose $l$-th coordinate have the largest absolute value.

From the slab partition, it is easy to see that 
for any $0\leq l\leq j$,  $\frac{|P\cap\Pi_l|}{|P|}=(1-\frac{\gamma}{j^2})^l\geq 1-\frac{\gamma}{j}$, and $\frac{|P\cap(\Pi_{l-1}\setminus\Pi_l) |}{|P|}\geq \frac{\gamma}{j^2}(1-\frac{\gamma}{j})$.  Thus, if  we set  $t=\frac{1}{\epsilon}$ and $\eta=\frac{\epsilon}{j}$, and use Lemma \ref{lem-select}  to take a random sample from $P$ of size $\frac{t\ln\frac{t}{\eta}}{1-\gamma/j}=\frac{\frac{1}{\epsilon}\ln\frac{j}{\epsilon^2}}{1-\gamma/j}$,  with probability $(1-\frac{\epsilon}{j})^j \geq 1-\epsilon$, the sample contains at least $\frac{1}{\epsilon}$ points from each $P\cap\Pi_l$. Also, if we set  $t=1$, and use Lemma \ref{lem-select} to take a random sample from $P$ of size $\frac{t\ln\frac{t}{\eta}}{(\gamma/j^2)(1-\gamma/j)}=\frac{\ln\frac{j}{\epsilon}}{(\gamma/j^2)(1-\gamma/j)}$, with probability $(1-\frac{\epsilon}{j})^j \geq 1-\epsilon$, the sample contains at least $1$ point from each $P\cap(\Pi_{l-1}\setminus\Pi_l)$. This means that if we take a 
random sample $\mathcal{S}$ of size $\frac{4j^2}{\gamma\epsilon}\ln\frac{j}{\epsilon}\geq \max\{\frac{\frac{1}{\epsilon}\ln\frac{j}{\epsilon^2}}{1-\gamma/j},\frac{\ln\frac{j}{\epsilon}}{(\gamma/j^2)(1-\gamma/j)}\}$, then 
with probability $(1-\epsilon)^2$, we have $\frac{4j^2}{\gamma\epsilon}\ln\frac{j}{\epsilon}\geq|A_l|\geq \frac{1}{\epsilon}$ and  $|A_{l}\setminus\Pi_{l}|\geq 1$ for any $l$ (by the fact that $|\mathcal{S}\cap P\cap(\Pi_{l-1}\setminus\Pi_l)|\geq 1$), where $A_l=\mathcal{S}\cap(P\cap\Pi_{l-1})$.

Let $l( p)$ be the $l$-th coordinate of a point $p$, 
$B_l=\{p\mid p\in-A_l\cup A_l, l(p)\geq 0\}$, and $\rho_l$  be the mean point of $B_l$ for $1\leq l\leq j$. Note that $\rho_l$ is contained in the output of Algorithm Symmetric-Sampling.

We define another sequence of slabs $\{\mathcal{R}'_1, \cdots, \mathcal{R}'_j\}$, where each $\mathcal{R}'_l$ is the slab axis parallel to $\mathcal{R}_l$ and with $\rho_l$ incident to one of its bounding hyperplanes. These slabs induce another slab partition on $\mathcal{F}$ with $\Pi'_l=\Pi'_{l-1}\cap\mathcal{R}'_l$. By Lemma \ref{lem-slab}, we know that there exists one point $\rho_{l_0}$ such that the slab determined by $\overrightarrow{o\rho}_{l_0}$ contains $\Pi'_j$ after amplified by a factor of $\sqrt{j}$. 

By the fact that $|A_l\setminus\Pi_{l}|\geq 1$ and the symmetric property of $-A_l\cup A_l$, we know that $|B_l\setminus\Pi_{l}|\geq 1$. Denote the width of $\Pi_l$ and $\Pi'_l$ in the $l$-th dimension as $a_l$ and $a'_l$ respectively. Then, $\frac{a'_l}{a_l}\geq \frac{1}{|B_l|}\geq\frac{\gamma\epsilon}{4j^2 \ln\frac{j}{\epsilon}} $ (note that $|B_l|=|A_l|$). 
Let $r=\frac{4j^2 \ln\frac{j}{\epsilon}}{\gamma\epsilon} $. Then,  $\Pi_l$ and $\Pi'_l$ differ in width (in the $l$-th dimension) by a factor no more than $r$. Thus, by the slab partition procedure, the difference of the width between $\Pi_j$ and $\Pi'_j$ is no more than a factor of $r$ in any direction on $\mathcal{F}$. As a result, the slab determined by $\overrightarrow{o\rho}_{l_0}$ contains  $\Pi_j$ after amplified by a factor of $\sqrt{j}r$.




Now, we come back to the case that $P$ locates in $\mathbb{R}^d$ rather than $\mathcal{F}$. First we let $P'=P\cap\Pi_j$.  Then, we have $|P'|\geq (1-\frac{\gamma}{j})|P|$. Note that we can always use the projection of $P$ whenever it is not in $\mathcal{F}$. Thus we can still select $P'$ by using slab partition on $\mathcal{F}$. This ensures the existence of $P'$. Next, we prove the existence of $\mathcal{F}'$.   

In this case, $\rho_{l_0}$ may not be in $\mathcal{F}$.  We will prove the rotation induced by $\overrightarrow{o\rho}_{l_0}$ is a $\Delta$-rotation of $\mathcal{F}$ with respect to $P\cap\Pi_j$. 
We let $Proj(\rho_{l_0})$ denote the projection of $\rho_{l_0}$ on $\mathcal{F}$, $\overrightarrow{u}=\frac{Proj(\rho_{l_0})-o}{||Proj(\rho_{l_0})-o||}$, and $h^2=\frac{1}{|P'|}\sum_{p\in P'}|<p, \overrightarrow{u}>|^2$. By the way how $\rho_{l_0}$ is generated,
we know that $||Proj(\rho_{l_0})-o||\geq \frac{1}{\sqrt{j}r}\max_{p\in P'}\{|<p, \overrightarrow{u}>|\}\geq \frac{1}{\sqrt{j}r} h$. Thus, in order to prove it is a $\frac{5\sqrt{j}r}{\sqrt{1-\gamma/j}}\delta_{P,\mathcal{F}}$-rotation, we just need to prove $||\rho_{l_0}-Proj(\rho_{l_0})||\leq \frac{5}{\sqrt{1-\gamma/j}}\delta_{P,\mathcal{F}}$. Recall that $\rho_{l_0}$ is the mean point of $B_{l_0}$, and $|B_{l_0}|=|A_{l_0}|\geq\frac{1}{\epsilon}$. Then we have the following claim (see Appendix for proof). 

\begin{claim}[1]
With probability $(1-\epsilon)^2$, $||\rho_{l_0}-Proj(\rho_{l_0})||\leq 5\delta_{P,\mathcal{F}}$.
\end{claim}

Claim (1) implies that $||\rho_{l_0}-Proj(\rho_{l_0})||\leq 5\delta_{P,\mathcal{F}}\leq \frac{5}{\sqrt{1-\gamma/j}}\delta_{P,\mathcal{F}}$, which means that the hyperplane $\frac{5\sqrt{j}r}{\sqrt{1-\gamma/j}}\delta_{P,\mathcal{F}}$-rotated from $\mathcal{F}$ is the desired $\mathcal{F}'$.
 
 As for success probability, since the success probability of containing $\rho_{l_0}$ is $(1-\epsilon)^2$ as shown in previous analysis, and the success probability for Claim (1) is also $(1-\epsilon)^2$, the success probability for Theorem \ref{the-alg1} is thus $(1-\epsilon)^4$.
\qed
\end{proof}

 \vspace{-0.3in}
\section{Approximation Algorithm for Projective Clustering}
\label{sec-general}
 \vspace{-0.1in}
This section presents a $(1+\epsilon)$-approximation algorithm for the projective clustering problem. For ease of understanding, we first give an outline of the algorithm.
 \vspace{-0.2in}
\subsection{Algorithm Outline}
 \label{sec-outline}
  \vspace{-0.1in} 
  
As mentioned in previous section, the objective of our approximation algorithm for the projective clustering problem is to determine $k$ $j$-dimensional flats such that $(1-\gamma)|P|$ of the input points in $P$ can be fit into the obtained $k$ flats, and the total objective value is no more than $(1+\epsilon)\delta^{2}_{opt}$, where $\delta^{2}_{opt}$ is the objective value of an optimal solution for all points in $P$.  Let   $\mathcal{C}=\{C_1, \cdots, C_k\}$ be the $k$ clusters in an optimal solution for $P$. 
 Our approach only consider those clusters (called {\em large clusters}) in $\mathcal{C}$ with size at least $\frac{\gamma}{2k}|P|$, since the union of the remaining clusters has a total size no more than $\frac{\gamma}{2}|P|$. Also, 
 for each large cluster $C_{i}$, our approach generates a flat to fit $(1-\frac{\gamma}{2})|C_{i}|$ of its points. Thus, in total we fit at least $(1-\frac{\gamma}{2})^2|P| \geq (1-\gamma)|P|$ of points to the resulting $k$ flats. 

Consider a large cluster $C_{i}$. 
Let $\mathcal{F}_i$ be its optimal fitting flat. 
It is easy to see that $\mathcal{F}_i$ passes through the mean point $o_{i}$ of $C_i$.
 Let $\delta^2_i=\frac{1}{|C_i|}\sum_{p\in C_i}||p, \mathcal{F}_i||^2$ be the optimal objective value of $C_{i}$. 
 To emulate the behavior of $C_{i}$, we first assume that we know the exact position of 
 $o_i$. If we run Algorithm Symmetric-Sampling on $o_i$ and a random sample of $C_{i}$ (note that since $C_{i}$ is a large cluster, by Lemma \ref{lem-select}, we can obtain enough points from $C_{i}$ by randomly sampling  $P$ directly), 
 by Theorem \ref{the-alg1}, we can get a point $s$, such that $\overrightarrow{o_i s}$ induces a $\Delta$-rotation for $\mathcal{F}_i$ with respect to a subset of $C_{i}$ with at least $(1-\frac{\gamma}{2j})|C_{i}|$ points, where $\Delta=\frac{5\sqrt{j}r}{\sqrt{1-\gamma/j}}\delta_i$. If we recursively run Algorithm Symmetric-Sampling $j$ times,  we obtain a sequence of $\Delta$-rotations and $j$ vectors which form a $j$-dimensional flat $\mathcal{F}'_i$ such that $\frac{1}{|C'_i|}\sum_{p\in C'_i}||p, \mathcal{F}'_i||^2\leq (1+5\sqrt{j}r)^j \delta^2_i$ (i.e., $\mathcal{F}'_i$ induces a $(1+5\sqrt{j}r)^j$-approximation for $C_{i}$), where $C'_i$ is a subset of $C_i$ with at least $(1-\frac{\gamma}{2j})^j|C_{i}|\ge 1-\frac{\gamma}{2}|C_{i}|$ points. 
Since Algorithm Symmetric-Sampling enumerates all subsets of $-\mathcal{S}\cup\mathcal{S}$, the above recursive procedure (called  {\em Algorithm Recursive-Projection}; see Section \ref{sec-rpa}.) forms  a hierarchical tree.  



From the proof of Theorem \ref{the-rpa}, we will know that the approximation ratio (i.e., $(1+5\sqrt{j}r)^{2j}$) is solely determined by the $\Delta$-rotation. In other words, if we could reduce the value of $\Delta$ from $\frac{5\sqrt{j}r}{\sqrt{1-\gamma/j}}\delta_i$ to $\frac{\epsilon}{2j}\delta_i$,  the approximation ratio would be reduced to $(1+\frac{\epsilon}{2j})^{2j}\approx 1+\epsilon$. 
To achieve this,  our idea is to find 
a point closer to $\mathcal{F}_i$ to induce the desired $\Delta$-rotation.
Our idea is to draw a ball $\mathcal{B}$ centered at every candidate point which induces the $\Delta$-rotation, and build a grid inside $\mathcal{B}$. The grid ensures the existence of one grid point close enough to $\mathcal{F}_i$. 
By using the linear time dimension reduction technique in \cite{DV07},  we can reduce the dimensionality of the projective clustering problem from $d$ to $(\frac{kj}{\epsilon})^{O(1)}$. This enables us to reduce the complexity of the grid. Thus, the approximation ratio can be reduced to 
 $1+\epsilon$ in linear time.

Now, the only remaining  issue is how to find the exact position of $o_i$. From Lemmas \ref{lem-dis} and \ref{lem-select}, we can find an approximate mean point  $\overline{o}_i$ for $o_i$. Thus, by translating $\mathcal{F}_i$ to pass though $\overline{o}_i$, we can show that $dist\{\mathcal{F}_i, \overline{\mathcal{F}}_i\}\leq O(\epsilon)\frac{1}{|C_i|}\sum_{p\in C_i}||p, \mathcal{F}_i||^2$. Combining this with the following Lemma \ref{lem-translate}, we can obtain a good approximation  $\overline{\mathcal{F}}_i$ for $\mathcal{F}_i$.


 \vspace{-0.1in}  
 \begin{lemma}
 \label{lem-translate}
 Let $P$ be a point set,  $\mathcal{F}$ be a $j$-dimensional flat, and $\overline{\mathcal{F}}$ be a translation of $\mathcal{F}$  in $\mathbb{R}^d$. Then,
 $\sqrt{\frac{1}{|P|}\sum_{p\in P}||p, \overline{\mathcal{F}}||^2}\leq \sqrt{\frac{1}{|P|}\sum_{p\in P}||p, \mathcal{F}||^2}+dist\{\mathcal{F}, \overline{\mathcal{F}}\}.$
 
 \end{lemma}
  \vspace{-0.1in} 

\vspace{-0.15in}
\subsection{Recursive Projection Algorithm}
\label{sec-rpa}
 \vspace{-0.05in}
  
 \noindent\textbf{Algorithm Recursive-Projection}
\newline \textbf{Input:} A point set $P$ and a single point $o\in \mathbb{R}^d$, $0<\gamma\leq 1$, and $1\leq j\leq d$.
\newline \textbf{Output:} A tree $\mathcal{T}$ of height $j$ with each node $v$ associated with a point $t_{v}$ and a flat $f_{v}$ in $\mathbb{R}^d$.
 \begin{enumerate} 
 
 \item Initialize $\mathcal{T}$ as a tree with a single root node associated with no point. The flat associated by the root is $\mathbb{R}^d$.
 
 \item Starting from the root, for each node $v$,  grow it in the following way
 
 \begin{enumerate}
 \item If the height of $v$ is $j$, it is a leaf node.
 
 \item Otherwise, 
 
 \begin{enumerate}
 
 \item Project $P$ onto $f_{v}$, and denote the projection as $\tilde{P}$. 
 
 \item Take a random sample $Q$ from $\tilde{P}$ with size $r=\frac{4j^2}{\gamma\epsilon}\ln\frac{j^2}{\epsilon}$, and run Algorithm Symmetric-Sampling on $Q$ and $o$ to obtain $2^{2r}$ points  $\overline{Q}$ as the output. 
 
 \item Create $2^{2r}$ children for $v$,  with each child associated with one point in $\overline{Q}$.
 
 \item Let $o$ be the mean point of $Q$. For each child $u$, let $t_{u}$ be its associated point;  associate $u$ the flat which is  the subspace of $f_{v}$ perpendicular to $\overrightarrow{ot_{u}}$ and with one less dimension than  $f_{v}$.
 
 \end{enumerate}
 
 \end{enumerate}
  
 \end{enumerate}
 \vspace{-0.05in}  
\noindent \textbf{Running Time}: It is easy to see that there are $O(2^{2rj})$ nodes in the output tree $\mathcal{T}$.  Each node costs $O(2^{2r}nd)$ time. Thus, the total running time  is $O(2^{2r(j+1)}nd)$.
  \vspace{-0.05in} 
\begin{theorem}
\label{the-rpa}
With probability $(1-\epsilon)^4$, the output $\mathcal{T}$ from Algorithm Recursive-Projection contains one root-to-leaf path such that the $j$ points associated with the path determine a flat $\mathcal{F}$ satisfying
inequality $\frac{1}{|P'|}\sum_{p\in P'}||p, \mathcal{F}||^2\leq (1+5\sqrt{j}r)^{2j}\frac{1}{|P|}\sum_{p\in P}||p, \mathcal{F}_{opt}||^2,$
where $P'$ is a subset of $P$ with at least $(1-\gamma)|P|$ points  and $\mathcal{F}_{opt}$ is the optimal fitting flat for $P$ among all $j$-dimensional flats passing through $o$.
\end{theorem}
\begin{proof}
Without loss of generality, we assume that $o$ is the origin of $\mathbb{R}^d$. Since all the related flats in the algorithm pass through $o$, we can view each flat as a subspace in $\mathbb{R}^d$. 

From the algorithm, we know that for any node $v$ at level $l$ of $\mathcal{T}$, $1\leq l\leq j$, there is a corresponding implicit point set $\tilde{P}$, which is the projection of $P$ on $f_v$. There is also an implicit flat $f_v\cap\mathcal{F}_{opt}$ in $f_v$. By Theorem \ref{the-alg1}, we know that there is one child of $v$, denoted as $v'$, such that the rotation for $f_v\cap\mathcal{F}_{opt}$ induced by $\overrightarrow{o t}_{v'}$ is a $\Delta$-rotation with respect to a subset $\tilde{P}'$ of  $\tilde{P}$ with size $|\tilde{P}'|\geq (1-\frac{\gamma}{j})|\tilde{P}|$, where $\Delta=\frac{5\sqrt{j}r}{\sqrt{1-\gamma/j}}\delta_{\tilde{P},f_{v\cap\mathcal{F}_{opt}}}$. 
Thus, if we always select such children (satisfying the above condition) from root to leaf,  we have a path with nodes $\{v_0, v_1, \cdots, v_j\}$, where $v_0$ is the root, and $v_{l+1}$ is the child of $v_l$. Correspondingly, a sequence of implicit point sets  $\{\tilde{P}_0, \tilde{P}_1, \cdots, \tilde{P}_j\}$ and a sequence of flats $\{\tilde{\mathcal{F}}_0, \tilde{\mathcal{F}}_1, \cdots, \tilde{\mathcal{F}}_j\}$
can also be obtained, which have the following properties.
 \vspace{-0.07in}
 \begin{enumerate}
\item Initially, $\tilde{\mathcal{F}}_0=\mathcal{F}_{opt}$, $\tilde{P}_0=P$.

\item For any $1\leq l\leq j$, $\tilde{\mathcal{F}}_l$ is the rotation of $\tilde{\mathcal{F}}_{l-1}\cap f_{v_{l-1}}$ induced by $\overrightarrow{o t}_{v_{l}}$,  $\tilde{P}_l$ is a subset of the projection of $\tilde{P}_{l-1}$ on $f_{v_{l-1}}$ with size at least $(1-\frac{\gamma}{j})|\tilde{P}_{l-1}|$, and  $\tilde{\mathcal{F}}_l$ is a $\Delta$-rotation with respect to $\tilde{P}_{l}$ (see Figure \ref{fig-project} ).
\end{enumerate}
 \vspace{-0.07in}
Note that since both $\tilde{\mathcal{F}}_l$ and $\tilde{\mathcal{F}}_{l-1}\cap f_{v_{l-1}}$ locate on $f_{v_{l-1}}$,  they are all perpendicular to $\overrightarrow{o t}_{v_{l-1}}$. The following claim reveals the
dimensionality of each $\tilde{\mathcal{F}}_l$ (see Appendix for proof).
\vspace{-0.1in}
  \begin{figure}[]
  \vspace{-0.18in}
  \centering
  \includegraphics[height=1.2in]{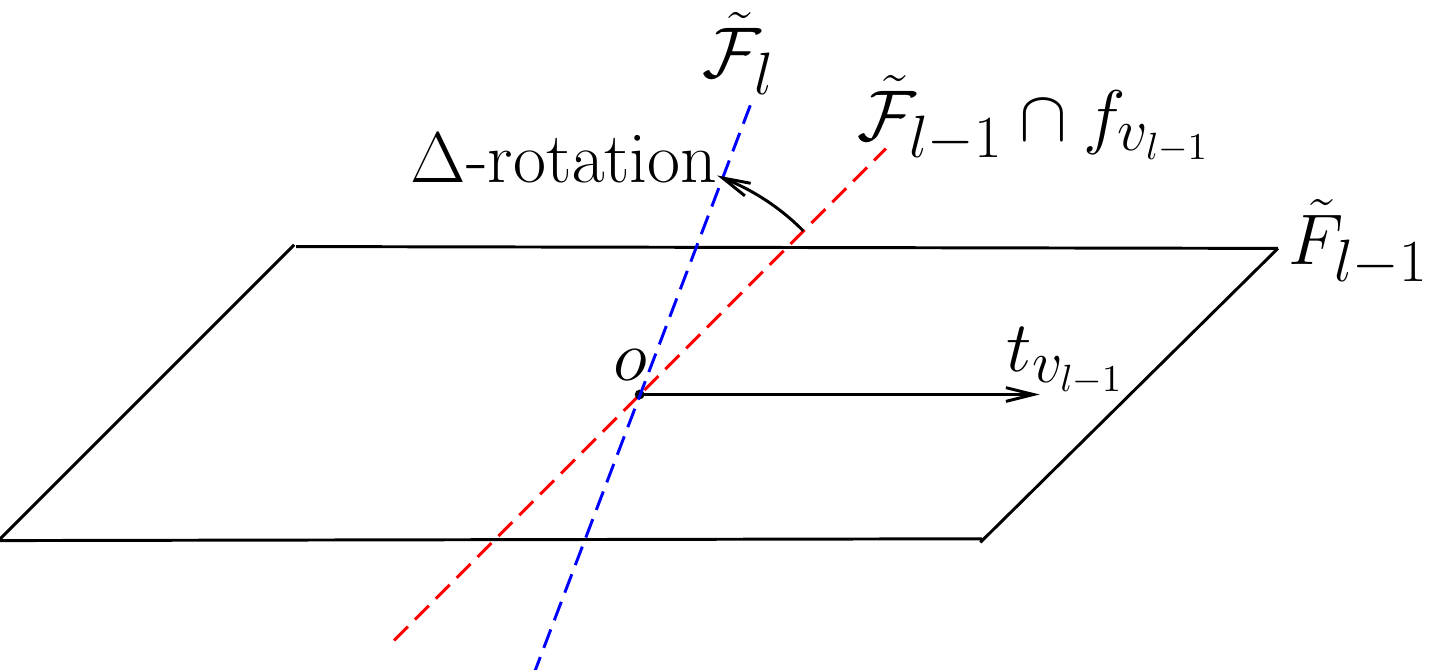}
  \vspace{-0.18in}
     \caption{An example illustrating how $\tilde{\mathcal{F}}_{l-1}$ evolves into $\tilde{\mathcal{F}}_l$.}
  \label{fig-project}
  \vspace{-0.18in}
  \end{figure}
\vspace{-0.2in}

\begin{claim}[\textbf{2}]
For $1\leq l\leq j$, $\tilde{\mathcal{F}}_l$ is a $(j-l+1)$-dimensional subspace.
\end{claim}
\vspace{-0.1in}
By Lemma \ref{lem-delta}, we can easily have the following claim.
\begin{claim}[\textbf{3}]
For any $1\leq l\leq j$, $\sum_{\tilde{p}\in\tilde{P}_l}||\tilde{p},\tilde{ \mathcal{F}}_l||^2\leq (1+5\sqrt{j}r)^2\sum_{\tilde{p}\in\tilde{P}_l}||\tilde{p}, \tilde{\mathcal{F}}_{l-1}\cap f_{v_{l-1}}||^2$.
\end{claim}
\vspace{-0.07in}
We construct as follows two other sequences, $\{P_1, \cdots, P_j\}$ and $\{\mathcal{F}_1,\cdots, \mathcal{F}_j \}$, for point sets and flats respectively.
\vspace{-0.1in}
\begin{enumerate}
\item Initially, $\mathcal{F}_1=\tilde{\mathcal{F}}_1$, $P_1=\tilde{P}_1$.
\item For any $2\leq l\leq j$, $\mathcal{F}_l=span\{\tilde{\mathcal{F}}_l, \overrightarrow{ot}_{v_1}, \cdots, \overrightarrow{o t}_{v_{l-1}}\}$, and $P_l$ is the corresponding point set of $\tilde{P}_{l}$ mapped back from $f_{v_{l}}$ to $\mathbb{R}^d$.
\end{enumerate}
\vspace{-0.08in}
From the above construction for $\mathcal{F}_l$, we have the following claim.
\vspace{-0.07in}
\begin{claim}[\textbf{4}]
For any point $\tilde{p}\in \tilde{P}_l$, let $p$ denote the corresponding point in $\mathbb{R}^d$. Then $||p, \mathcal{F}_l||=||\tilde{p}, \tilde{\mathcal{F}}_l||$. 
\end{claim}
 \vspace{-0.03in}
From Claim (\textbf{2}), we know that each $\mathcal{F}_l$ is a $j$-dimensional subspace. From the algorithm, we know that $span\{\tilde{\mathcal{F}}_{l-1}\cap f_{v_{l-1}}, \overrightarrow{o t}_{v_{l-1}}\}=\tilde{\mathcal{F}}_{l-1}$ (see Fig. \ref{fig-project}), which implies $span\{\tilde{\mathcal{F}}_{l-1}\cap f_{v_{l-1}}, \overrightarrow{ot}_{v_1}, \cdots, $ $\overrightarrow{o t}_{v_{l-1}}\}=\mathcal{F}_{l-1}$. Then by Claim (\textbf{3}) and (\textbf{4}), we have the following inequality
\small{
\begin{eqnarray}
\sum_{p\in P_l}||p, \mathcal{F}_l||^2\leq (1+5\sqrt{j}r)^2\sum_{p\in P_l}||p, \mathcal{F}_{l-1}||^2.\label{for-9}
\end{eqnarray}
}
\normalsize
 \vspace{-0.07in}
By the definition of $P_l$, we have $P_j\subset P_{j-1}\subset \cdots \subset P_1\subset P_0=P$. Thus,
\small{
\begin{eqnarray}
 \sum_{p\in P_{l}}||p, \mathcal{F}_{l-1}||^2\leq\sum_{p\in P_{l-1}}||p, \mathcal{F}_{l-1}||^2. \label{for-10}
 \end{eqnarray}
 }
 \normalsize
  \vspace{-0.1in}
 Combining (\ref{for-9}) and (\ref{for-10}), we have
 \vspace{-0.02in} 
 \small{
 \begin{eqnarray} 
 \sum_{p\in P_l}||p, \mathcal{F}_l||^2\leq (1+5\sqrt{j}r)^2\sum_{p\in P_{l-1}}||p, \mathcal{F}_{l-1}||^2.
 \end{eqnarray}
 }
 \normalsize
  \vspace{-0.15in}
 Recursively using the above inequality, we have $\sum_{p\in P_j}||p, \mathcal{F}_j||^2\leq (1+5\sqrt{j}r)^{2j}\sum_{p\in P}||p, \mathcal{F}_{opt}||^2$. From the definition of $P_j$, we know that $|P_j|\geq (1-\frac{\gamma}{j})^j |P|\geq (1-\gamma)|P|$. Furthermore, by Claim (\textbf{2}), we know that $\tilde{\mathcal{F}}_j$ is a $1$-dimension flat (i.e., the single line spanned by $\overrightarrow{o t}_{v_{j}}$), hence $\mathcal{F}_j=span\{\overrightarrow{ot}_{v_1}, \cdots, \overrightarrow{o t}_{v_{j}}\}$. Thus, if setting $P'=P_j$, and $\mathcal{F}=\mathcal{F}_j$, we have $\frac{1}{|P'|}\sum_{p\in P'}||p, \mathcal{F}||^2\leq (1+5\sqrt{j}r)^{2j}\frac{1}{|P|}\sum_{p\in P}||p, \mathcal{F}_{opt}||^2$. 
 
\noindent\textbf{Success probability:} Each time we use Theorem \ref{the-alg1}, the success probability is $(1-\frac{\epsilon}{j})^4$ (we replace $\epsilon$ by $\frac{\epsilon}{j}$ to increase the sample size). Thus, the total success probability is $(1-\frac{\epsilon}{j})^{4j}\geq (1-\epsilon)^4$.
\qed
\end{proof}

\vspace{-0.23in}
\subsection{Algorithm for Projective Clustering}
 \label{sec-pc}
\vspace{-0.05in}
 \noindent\textbf{Algorithm Projective-Clustering}
\newline \textbf{Input:} A set of points  $P$ in $\mathbb{R}^d$, positive integers $k$, $j$, and two positive numbers $0<\epsilon, \gamma \leq 1$. 
\newline \textbf{Output:} An approximate solution for $(k,j)$-projective clustering
\begin{enumerate}

\item Use the dimension reduction technique in \cite{DV07} to reduce the dimensionality from $d$ to $d'=(\frac{kj}{\epsilon})^{O(1)}$.

\item Set the sample size $r=\frac{2kt}{\gamma}\ln(2kt)$, where $t=\frac{8j^2}{\gamma\epsilon}\ln\frac{kj^2}{\epsilon}$.
\item Running Algorithm Recursive-Project  $k$ time with sample size $r$. Denote the $k$ output trees as $\{\mathcal{T}_1, \cdots, \mathcal{T}_k\}$.

\item Enumerate the combinations of all $k$ flats from the $k$ trees: Select one flat yielded by a root-to-leaf path from each $\mathcal{T}_l$ for $1\leq l\leq k$, and compute the objective value of these $k$ flats.  Let $\mathcal{L}$ be the smallest objective value among all the combinations.
\item Re-run Algorithm Recursive-Projection $k$ time with the following modification: In Step $2(b)$, for each point $p$ of the $2^{2r}$ points returned by  Algorithm Symmetric-Sampling, build a ball $\mathcal{B}$ centered at $p$ and with radius $r_{\mathcal{B}}$, and construct a grid inside $\mathcal{B}$. For each grid point, create a node, associate it  with the grid point, and make it as a sibling of the node containing $p$.   


\item Enumerate the combinations of all $k$ flats from the $k$ output trees of the above step. Find the $k$ flats with the smallest objective value for $P$ and output them as the solution.

\end{enumerate}

In the above algorithm, the radius $r_{\mathcal{B}}$ and the density of the grid are chosen in a way so that there exists a grid point which induces a $\frac{\epsilon}{4j}$-rotation for the clustering points. Thus, we can further reduce the approximation ratio to $(1+\epsilon)$, and have the following theorem.
Detailed analysis on the algorithm and the theorem  is left in Section \ref{sec-parameter} of the Appendix.
\vspace{-0.05in}
\begin{theorem}
\label{the-ptas}
Let $P$ be a set of $\mathbb{R}^{d}$ points in a $(k,j)$-projective clustering instance. Let $Opt$ be the optimal objective value on $P$. With constant probability and in $O(2^{poly(\frac{kj}{\epsilon\gamma})}nd)$ time, Algorithm Projective-Clustering outputs an approximate solution $\{\mathcal{F}_1, \cdots, \mathcal{F}_k\}$ such that each $\mathcal{F}_l$ is a $j$-flat, and $\frac{1}{|P'|}\sum_{p\in P'}\min_{1\leq l\leq k}||p, \mathcal{F}_l||^2\leq (1+\epsilon) Opt$, where $P'$ is a subset of $P$ with at least $(1-\gamma)|P|$ points. 
\end{theorem}

\vspace{-0.15in}
\subsection{Extensions.}
\vspace{-0.05in}

We present two main extensions. See Appendix for details. 

\textbf{Linear time PTAS for regular projective clustering:} For points with bounded Coefficient of Variation (CV) (we call such problem as regular projective clustering), we show that our approach leads to a linear time PTAS solution.  
The main idea is that since the CV is bounded, the point from symmetric sampling algorithm is far enough to $o$. This implies that the each $\Delta$-rotation from Algorithm Recursive-Projection is for the whole set $P$ rather than a subset $P'$. Thus we can fit all points of $P$ into the resulting $k$ flats within the same approximate ratio.  

\textbf{$L_\tau$ sense Projective clustering:} We show that our approach can be  extended to $L_\tau$ sense projective clustering for any integer $1\leq \tau<\infty$ and achieve similar results for both general and regular projective clustering. The key idea is to define the $L_\tau$ sense $\Delta$-rotation, and prove a result similar to Lemma \ref{lem-delta}. In other words, the symmetric sampling technique can also yield a $L_\tau$ sense $\Delta$-rotation  with 

\vspace{-0.1in}
\small{
$$(\frac{1}{|P|}\sum_{p\in P}||p, \mathcal{F}'||^\tau)^{1/\tau}\leq(\frac{1}{|P|}\sum_{p\in P}||p, \mathcal{F}||^\tau)^{1/\tau}+\Delta.$$
}
\normalsize



%
%
%
%
%
%

\newpage

\section{Proof of Lemma \ref{lem-select} }
\begin{proof}
If we randomly select $z$ elements from $S$, then it is easy to know that with probability $1-(1-\alpha)^z$, there is at least one element from the sample belonging to $S'$. If we want the probability  $1-(1-\alpha)^z$  equal to $1-\eta/t$, $z$ has to be $\frac{\ln\frac{t}{\eta}}{\ln\frac{1}{1-\alpha}}=\frac{\ln\frac{t}{\eta}}{\ln(1+\frac{\alpha}{1-\alpha})}\leq\frac{\ln\frac{t}{\eta}}{\ln(1+\alpha)}$.

Thus if we perform $t$ rounds of random sampling with each round selecting $\frac{\ln\frac{t}{\eta}}{\ln(1+\alpha)}$ elements,  we get at least $t$ elements from $S'$ with probability at least $(1-\eta/t)^t\geq 1-\eta$.
\qed
\end{proof}

\section{Proof for Claim 1 in Theorem \ref{the-alg1}}
\begin{proof}[\textbf{of Claim (1)}]
We first reduce the space from $\mathbb{R}^d$ to $\mathcal{F}^\bot$, which is a $(d-j)$-dimensional subspace. For simplicity, we use the same notations for points in $\mathcal{F}^\bot$ as in $\mathbb{R}^d$. It is easy to know that during the space reduction from $\mathbb{R}^d$ to $\mathcal{F}^\bot$, $Proj(\rho_{l_0})$ is projected to the origin, and $||\rho_{l_0}-Proj(\rho_{l_0})||$ is equal to $||\rho_{l_0}-o||$ in the subspace. We let $B^1_{l_0}=B_{l_0}\cap -\mathcal{S}$, and $B^2_{l_0}=B_{l_0}\cap \mathcal{S}$. Correspondingly, we let $-B^1_{l_0}$ denote the symmetric point set of $B^1_{l_0}$ with respect to $o$. Without loss of generality, we assume that $|B^1_{l_0}|\geq |B^2_{l_0}|$. 
Let  $s_1$ and $s_2$ be the mean points of $B^1_{l_0}$ and $B^2_{l_0}$ respectively, and $\alpha_1=\frac{|B^1_{l_0}|}{|B_{l_0}|}$ and $\alpha_2=\frac{|B^2_{l_0}|}{|B_{l_0}|}$.  Since $\rho_{l_0}$ is the mean point of $B_{l_0}$,  we have $\rho_{l_0}=\alpha_1 s_1+\alpha_2 s_2$.  Thus,
\begin{eqnarray}
||\rho_{l_0}-o||&=&||\alpha_1 s_1+\alpha_2 s_2-o||=||\alpha_1(2o-s_1)+\alpha_2 s_2-2\alpha_1 o+2\alpha_1 s_1 -o|| \nonumber\\
&=&||\alpha_1(2o-s_1)+\alpha_2 s_2-o-2\alpha_1(2o-s_1-o)|| \nonumber\\
&\leq& ||\alpha_1(2o-s_1)+\alpha_2 s_2-o||+2\alpha_1||2o-s_1-o||, \label{for-7}
\end{eqnarray}
where the last inequality follows from triangle inequality. 

Note that $2o-s_1$ is the mean point of $-B^1_{l_0}$,  $\alpha_1(2o-s_1)+\alpha_2 s_2$ is the mean point of $-B^1_{l_0}\cup B^2_{l_0}$, and  $-B^1_{l_0}\cup B^2_{l_0}$ is a sample from $P$ of size at least $\frac{1}{\epsilon}$. Let $\pi$ be the mean point of $P$, and $a^2=\frac{1}{|P|}\sum_{p\in P}||p-\pi||^2$ (note that the current space is $\mathcal{F}^\bot$). Then by Lemma \ref{lem-dis}, we know that with probability $1-\eta$, $||\alpha_1(2o-s_1)+\alpha_2 s_2-\pi||\leq \sqrt{\epsilon/\eta}a$. Similarly, since $-B^1_{l_0}$ is a sample from $P$ with size at least $\frac{1}{2\epsilon}$ (by $|B^1_{l_0}|\geq |B^2_{l_0}|$), we have $||2o-s_1-\pi||\leq \sqrt{2\epsilon/\eta}a$ with probability $1-\eta$. Since, $\delta^2_{P,\mathcal{F}}=\frac{1}{|P|}\sum_{p\in P}||p-o||^2=a^2+||\pi-o||^2$,  with a total probability $(1-\eta)(1-\eta)$, we have
\begin{eqnarray}
&&||\alpha_1(2o-s_1)+\alpha_2 s_2-o||+2\alpha_1||2o-s_1-o|| \nonumber\\
&\leq& ||\alpha_1(2o-s_1)+\alpha_2 s_2-\pi||+||\pi-o||+2\alpha_1(||2o-s_1-\pi||+||\pi-o||) \nonumber\\
&\leq& \sqrt{\epsilon/\eta}a+2\alpha_1\sqrt{2\epsilon/\eta}a+(1+2\alpha_1)||\pi-o|| \nonumber\\
&=&(1+2\sqrt{2})\sqrt{\epsilon/\eta}a+3||\pi-o|| \nonumber\\
&\leq& \sqrt{(1+2\sqrt{2})^2\epsilon/\eta+9}\sqrt{a^2+||\pi-o||^2}=\sqrt{(1+2\sqrt{2})^2\epsilon/\eta+9}\delta_{P,\mathcal{F}}, \label{for-8}
\end{eqnarray}

where the first inequality follows from triangle inequality, and the last inequality follows from the fact that $x_1y_1+x_2y_2\leq \sqrt{x^2_1+x^2_2}\sqrt{y^2_1+y^2_2}$ for any four real numbers $x_1,x_2,y_1,y_2$. 

Setting $\eta=\epsilon$,  by (\ref{for-7}) and (\ref{for-8}), we have $||\rho_{l_0}-o||\leq\sqrt{18+4\sqrt{2}}\delta<5\delta_{P,\mathcal{F}} $, with probability $(1-\epsilon)^2$.
\qed
\end{proof}

\section{Proof of Claim 2 in Theorem \ref{the-rpa}}
\begin{proof}
We prove this claim by induction. For the base case (i.e., $l=1$), $\tilde{\mathcal{F}}_1$ has the same dimensionality as $\tilde{\mathcal{F}}_0\cap f_{v_0}=\mathcal{F}_{opt}$ (note that $f_{v_0}=\mathbb{R}^d$). Hence $\tilde{\mathcal{F}}_1$ is a $j$-dimensional subspace. Then we assume that $\tilde{\mathcal{F}}_w$ is a $(j-w+1)$-dimensional subspace for any $w\leq l$ (i.e., induction hypothesis). Now, we consider the case of $l+1$. Since $\tilde{\mathcal{F}}_{l+1}$ is only a rotation of $\tilde{\mathcal{F}}_{l}\cap f_{v_{l}}$, they have the same dimensionality. Also, from the algorithm, we know that $\tilde{\mathcal{F}}_{l}\cap f_{v_{l}}$ is the subspace in $\tilde{\mathcal{F}}_{l}$ which is perpendicular to $Proj(\overrightarrow{o t}_{v_{l}})$, where $Proj(\overrightarrow{ot}_{v_{l}})$ is the projection of $\overrightarrow{ot}_{v_{l}}$ on $\tilde{\mathcal{F}}_{l}$. Thus, $\tilde{\mathcal{F}}_{l}\cap f_{v_{l}}$ is a $(j-l+1-1)$-dimensional subspace, which implies that $\tilde{\mathcal{F}}_{l+1}$ is also a $(j-l)$-dimensional subspace. Hence, Claim (\textbf{2}) is proved.
\qed
\end{proof}

\section{Proof of Lemma \ref{lem-translate}}

 \begin{proof}
For simplicity, we let $||\mathcal{F}, \overline{\mathcal{F}}||$ denote $dist\{\mathcal{F}, \overline{\mathcal{F}}\}$. By triangle inequality, for any $p\in P$, we have $||p, \overline{\mathcal{F}}||\leq ||p, \mathcal{F}||+||\mathcal{F}, \overline{\mathcal{F}}||$. Thus, 
\begin{eqnarray} 
 \frac{1}{|P|}\sum_{p\in P}||p, \overline{\mathcal{F}}||^2 &\leq& \frac{1}{|P|}\sum_{p\in P}( ||p, \mathcal{F}||+||\mathcal{F}, \overline{\mathcal{F}}||)^2 \nonumber \\
 &=&\frac{1}{|P|}\sum_{p\in P}(||p, \mathcal{F}||^2+2||p, \mathcal{F}|| \times ||\mathcal{F},\overline{\mathcal{F}}||+||\mathcal{F}, \overline{\mathcal{F}}||^2). \label{for-11}
 \end{eqnarray}
 
Let $c=\frac{||\mathcal{F}, \overline{\mathcal{F}}||}{\sqrt{\frac{1}{|P|}\sum_{p\in P}||p, \mathcal{F}||^2}}$. Then, we have
\begin{eqnarray}
2||p, \mathcal{F}||\times ||\mathcal{F}, \overline{\mathcal{F}}||\leq c(||p, \mathcal{F}||^2+\frac{1}{|P|}\sum_{p\in P}||p, \mathcal{F}||^2). \label{for-12}
\end{eqnarray}

Combining (\ref{for-11}) and (\ref{for-12}), we have 
\begin{eqnarray*}
 \frac{1}{|P|}\sum_{p\in P}||p, \mathcal{F}'||^2&\leq&\frac{1}{|P|}\sum_{p\in P}(||p, \mathcal{F}||^2+c(||p, \mathcal{F}||^2+\frac{1}{|P|}\sum_{p\in P}||p, \mathcal{F}||^2)+||\mathcal{F}, \overline{\mathcal{F}}||^2)\\
 &=&\frac{1+2c}{|P|}\sum_{p\in P}||p, \mathcal{F}||^2+||\mathcal{F}, \overline{\mathcal{F}}||^2=(1+c)^2\frac{1}{|P|}\sum_{p\in P}||p, \mathcal{F}||^2\\
  &=&(\sqrt{\frac{1}{|P|}\sum_{p\in P}||p, \mathcal{F}||^2}+||\mathcal{F}, \overline{\mathcal{F}}||)^2.
\end{eqnarray*}
 Thus, the lemma is true.
 \qed
 \end{proof}

\section{Some Details Analysis for Algorithm Projective-Clustering}
\label{sec-parameter}

\noindent\textbf{ Sample size and success probability:}  Theorem \ref{the-rpa} enables us to find a good flat for the whole point set $P$. To find $k$ good flats, one for each cluster, 
we need to increase the probability from $(1-\epsilon)^4$ to $(1-\frac{\epsilon}{k})^4$, and replace $\gamma$ by $\frac{\gamma}{2}$. Thus, for each cluster, we need 
$t=\frac{8j^2}{\gamma\epsilon}\ln\frac{kj^2}{\epsilon}$ points. By Lemma \ref{lem-select}, we know that if we want a random sample containing at least $t$ points from one cluster with probability $1-\frac{1}{2k}$, we need to sample $r=\frac{t\ln2kt}{\alpha}$ points  from $P$, where $\alpha$ is the fraction of the cluster in $P$. Note that since our algorithm only focuses on emulating the behavior of large clusters, $\alpha\geq\frac{\gamma}{2k}$. This means that  it is sufficient to set the sample size $r$ to be $\frac{2kt}{\gamma}\ln 2kt$. The total success probability is therefore $((1-\frac{1}{2k})(1-\frac{\epsilon}{k})^4)^k\geq \frac{1}{2}(1-\epsilon)^4$.

\noindent\textbf{Radius and grid density:} Let $Opt$ be the optimal objective value of  projective clustering  on $P$. By Theorem \ref{the-rpa}, we know that Step $4$ in the algorithm outputs an objective value $\mathcal{L}\leq (1+5\sqrt{j}r)^{2j} Opt$. From the proof of Theorem \ref{the-rpa}, we know that  on the path generating the resulting flat, 
if each node incurs  a $\frac{\epsilon}{4j}$-rotation rather than a $\Delta$-rotation, the approximation ratio will become $(1+\frac{\epsilon}{4j})^{2j}\leq 1+\epsilon$ (instead of $(1+5\sqrt{j}r)^{2j}$). Thus, to reduce $5\sqrt{j}r$ to $\frac{\epsilon}{4j}$, we can build a grid around the point associated with each node $v$ in the tree $\mathcal{T}$ generated by Algorithm Recursive-Projection.  For each grid point, add a  
node as a new sibling of $v$ (see Step $5$) and associate it with the grid point.  The problem is how to determine the density of the grid so as to generate the desired approximation ratio.

To determine the density of the grid, we first have
\begin{eqnarray}
 1\leq \frac{\sqrt{\mathcal{L}}}{\sqrt{Opt}}\leq (1+5\sqrt{j}r)^j .\label{for-6}
\end{eqnarray}

We use the same notations as in Theorem \ref{the-rpa}. Let $v_{l}$ be the current node in Algorithm Recursive-Projection, and $Proj(t_{v_{l}}$ be the projection of $t_{v_l}$ on $f_{v_l}\cap\mathcal{F}_{opt}$.
Further, let $h^2=\frac{1}{|\tilde{P}_l|}\sum_{p\in\tilde{P}_l}||<p, Proj(t_{v_l})>||^2$, and $h'=||o-Proj(t_{v_l})||$. By Definition \ref{def-delta}, we have
\begin{eqnarray}
 \frac{||t_{v_l}-Proj(t_{v_l}) ||/h'}{\sqrt{Opt}/h}\leq 5\sqrt{j}r. \label{for-27}
\end{eqnarray}

Combining (\ref{for-6}) and (\ref{for-27}), we have
\begin{eqnarray}
\frac{5\sqrt{j}r\frac{h'}{h}\sqrt{\mathcal{L}}}{||t_{v_l}-Proj(t_{v_l}) ||}\geq 1.\label{for-28}
\end{eqnarray}
By Lemma \ref{lem-slab}, we know that $\frac{1}{t}\leq \frac{h'}{h}\leq 1$ (note that $t=\frac{8j^2}{\gamma\epsilon}\ln\frac{kj^2}{\epsilon}$, as discussed previously). Hence, we can set the radius $r_{\mathcal{B}}$ of $\mathcal{B}$ to be $5\sqrt{j}r\sqrt{\mathcal{L}}$. Thus, $r_{\mathcal{B}}>||t_{v_l}-Proj(t_{v_l}) ||$,  which implies that $Proj(t_{v_l})$ locates inside $\mathcal{B}$. If we set $\epsilon_0=\frac{\epsilon}{4j(1+5\sqrt{j}r_{\mathcal{B}})^j\lambda t}$, and construct a grid inside $\mathcal{B}$ with grid length $\frac{\epsilon_0}{\sqrt{d}} r_{\mathcal{B}}$, then there is one grid point $\pi$ satisfying the following inequality.
\begin{eqnarray}
||\pi-Proj(t_{v_l})||&\leq&\sqrt{d(\frac{\epsilon_0}{\sqrt{d}} r_{\mathcal{B}})^2}=\epsilon_0 r_{\mathcal{B}}=\frac{\epsilon5\sqrt{j}r\sqrt{\mathcal{L}}}{4j(1+5\sqrt{j}r)^j\lambda t} \nonumber \\
&\leq& \frac{\epsilon}{4jt}\sqrt{Opt}, \label{for-30}
\end{eqnarray}
 where the last inequality follows from (\ref{for-6}). Combining inequalities (\ref{for-30}) and $\frac{1}{t}\leq \frac{h'}{h}$, we have $\frac{||\pi-Proj(t_{v_l})||}{h'}\leq \frac{\epsilon}{4j}\frac{\sqrt{Opt}}{h}$, which implies that it  induces a $\frac{\epsilon}{4j}$-rotation. 
 
\noindent\textbf{Running time.} The dimension reduction step costs $O(nd\cdot$poly$(\frac{kj}{\epsilon}))$ time, and resulting problem has diemsion $d'=$poly$(\frac{kj}{\epsilon})$. Step $3$ and $4$ take $O(k2^{2rkj}nd')$ time. Note that in Step $5$, the complexity of the grid inside $\mathcal{B}$ is $(\frac{\sqrt{d'}}{\epsilon_0})^{d'}$. Thus, the complexity of each $\mathcal{T}_l$ will increase to $O((2^{2r}(\frac{\sqrt{d'}}{\epsilon_0})^{d'})^j)$. Hence, the total running time is $2^{poly(\frac{kj}{\epsilon\gamma})}nd$.

\section{PTAS for Regular Projective Clustering}
\label{sec-regular}

This section first introduces the {\em regular}  projective clustering problem, and then presents a PTAS for it. 
We start our discussion with  a concept used in statistics.

\begin{definition}[Coefficient of Variation (CV)]
\label{def-cv}
Let $x$ be a random variable, and $\mu=E[x]$ be its expection. The coefficient of variation of $x$ is denoted as $\frac{\sqrt{E[(x-\mu)^2]}}{E[|x-\mu|]}$.
\end{definition}

CV of a single variable aims to measure the dispersion of the variable in a way that does not depend on the variable's actual value. 
The higher the CV, the greater the dispersion is in the variable. Distributions with $CV<1$ (such as an Erlang distribution) are considered as low-variance, while those with $CV>1$ (such as a hyper-exponential distribution) are considered as high-variance. Note that many commonly encountered distributions has constant CV (e.g., Gaussian distribution \footnote{Let $f(x)=\frac{1}{\sqrt{2\pi\delta^2}}e^{-\frac{x^2}{2\delta^2}}$ be any Gaussian distribution with mean point at the origin. Then, $E[|x|]=\int_{-\infty}^{+\infty}|x|\frac{1}{\sqrt{2\pi\delta^2}}e^{-\frac{x^2}{2\delta^2}} \,dx =2\int_{0}^{+\infty}
x\frac{1}{\sqrt{2\pi\delta^2}}e^{-\frac{x^2}{2\delta^2}}\, dx=\frac{2\delta}{\sqrt{2\pi}}.$ Since $E[x^2]=\delta^2$, we have $CV=\frac{\sqrt{E[x^2]}}{E[|x|]}=\sqrt{\frac{\pi}{2}}$.}).

\begin{lemma}
\label{lem-bound}
Let $\mathcal{X}=\{x_{1}, \cdots, x_{n}\}$ be a set of $n$ numbers with  coefficient of variation $\omega$, and $S=\{x_{i_{1}}, \cdots, x_{i_{m}}\}$ be a random sample of $\mathcal{X}$. 
Then, for any positve constant $\eta$, $Prob(\frac{\sum^{m}_{l=1}|x_{i_l}-\mu|}{m}\ge (1-\eta\sqrt{\frac{\omega^2-1}{m}})\frac{h}{\omega})\ge 1-\frac{1}{\eta^2},$ where
$\mu=\frac{1}{n}\sum^n_{i=1}x_i$ and $h^2=\frac{1}{n}\sum^n_{i=1}(x_i-\mu)^2$.
\end{lemma}
\begin{proof}
Let $\tilde{\mu}=\frac{1}{n}\sum^n_{i=1}|x_{i}-\mu|$. Then, from the definition of CV, we know that $\omega=\frac{\sqrt{\frac{1}{n}\sum^n_{i=1}(x_i-\mu)^2}}{\frac{1}{n}\sum^n_{i=1}|x_{i}-\mu|}=\frac{h}{\tilde{\mu}}$, which implies $h=\omega\tilde{\mu}$.

Since the variance of $\{|x_{i}-\mu| \mid 1\leq i\leq n\}$ is $\frac{1}{n}\sum^n_{i=1}(|x_i-\mu|-\tilde{\mu})^2=h^2-\tilde{\mu}^2$, and $S$ is the random sample from $\mathcal{X}$ with size $m$, we know that the expected value and variance of $\frac{\sum^{m}_{l=1}|x_{i_l}-\mu|}{m}$ are $\tilde{\mu}$ and $\frac{h^2-\tilde{\mu}^2}{m}$ respectively. By Markov inequality, we know
\begin{eqnarray}
Prob(|\frac{\sum^{m}_{l=1}|x_{i_l}-\mu|}{m}-\tilde{\mu}|\ge \eta\sqrt{\frac{h^2-\tilde{\mu}^2}{m}})\leq \frac{1}{\eta^2}. \label{for-33}
\end{eqnarray}
Meanwhile, since 
\begin{eqnarray*}
|\frac{\sum^{m}_{l=1}|x_{i_l}-\mu|}{m}-\tilde{\mu}|\le \eta\sqrt{\frac{h^2-\tilde{\mu}^2}{m}}\Longrightarrow \frac{\sum^{m}_{l=1}|x_{i_l}-\mu|}{m}\ge\tilde{\mu}-\eta\sqrt{\frac{h^2-\tilde{\mu}^2}{m}},
\end{eqnarray*}
we have
\begin{eqnarray}
Prob(\frac{\sum^{m}_{l=1}|x_{i_l}-\mu|}{m}\ge\tilde{\mu}-\eta\sqrt{\frac{h^2-\tilde{\mu}^2}{m}})\ge Prob(|\frac{\sum^{m}_{l=1}|x_{i_l}-\mu|}{m}-\tilde{\mu}|\le \eta\sqrt{\frac{h^2-\tilde{\mu}^2}{m}}). \label{for-34}
\end{eqnarray}

Combining (\ref{for-33}) and (\ref{for-34}), we have
\begin{eqnarray*}
Prob(\frac{\sum^{m}_{l=1}|x_{i_l}-\mu|}{m}\ge\tilde{\mu}-\eta\sqrt{\frac{h^2-\tilde{\mu}^2}{m}})\ge 1-\frac{1}{\eta^2}.
\end{eqnarray*}

Recall that $h=\omega\tilde{\mu}$. If we replace $\tilde{\mu}$ by $\frac{h}{\omega}$,  the above inequality becomes
\begin{eqnarray*}
Prob(\frac{\sum^{m}_{l=1}|x_{i_l}-\mu|}{m}\ge (1-\eta\sqrt{\frac{\omega^2-1}{m}})\frac{h}{\omega})\ge 1-\frac{1}{\eta^2}.
\end{eqnarray*}
\qed
\end{proof}

Using coefficient of variation, we introduce the  regular projective clustering problem.

Let $P$ be a point set in $\mathbb{R}^d$, $\mathcal{F}_{opt}$ be its optimal $j$-dimensional flat fitting, and $o$ be its mean point. 
It is easy to see that $o$ locates on $\mathcal{F}_{opt}$.

\begin{definition}[Regular Single $j$-Flat Fitting]
\label{def-rff}
A single $j$-flat fitting problem with input point set $P$ is regular if for any direction 
 $\overrightarrow{v}$, 
 the coefficient of variation of $\{<\overrightarrow{op}, \overrightarrow{v}>\mid p\in P\}$ is bounded by some constant $\omega$. $\omega$ is called the regular factor of $P$.
\end{definition}


\begin{definition}[Regular $(k,j)$-Projective Clustering]
\label{def-rpc}
A $(k,j)$-projective clustering  problem with input point set $P$ and  optimal clusters $\{C_{1}, \cdots, C_{k}\}$  is regular  if each $C_{i}$ is a regular  single $j$-flat fitting problem.
\end{definition}

%
%
%
%
%
%

The following theorem is a counterpart of Theorem \ref{the-alg1} for Regular projective clustering. 
  
\begin{theorem}
\label{the-ralg1}
Let $P=\{p_1, \cdots, p_n\}$ be the $\mathbb{R}^{d}$ input point set of a regular single $j$-flat fitting problem with mean point $o$ and regular factor $\omega$,  and $\mathcal{S}$ be its random sample of size $m=\frac{\omega^2-1}{\epsilon^2}$, where $\epsilon>0$ is a small constant. Let  $\mathcal{F}_{opt}$ be a $j$-dimensional flat, and $\overline{\mathcal{S}}$ is the set of points returned by Algorithm Symmetric-Sampling on $\mathcal{S}$ and $o$. 
Then with probability $(1-\epsilon)(1-\frac{\epsilon^2}{\omega^2-1})^2$, $\overline{\mathcal{S}}$ contains one point $s$ such that  the flat $\mathcal{F}'$ rotated from $\mathcal{F}_{opt}$ and induced by $\overrightarrow{os}$ satisfies the following inequality,
$$\sum_{p\in P}||p, \mathcal{F}'||^2\leq (1+\frac{5\sqrt{j}\omega}{1-\sqrt{\epsilon}})^2\sum_{p\in P}||p, \mathcal{F}_{opt}||^2.$$
  \end{theorem}
  
  \vspace{-0.15in}
     \begin{figure}[]
     \vspace{-0.15in}
  \centering
  \includegraphics[height=1.8in]{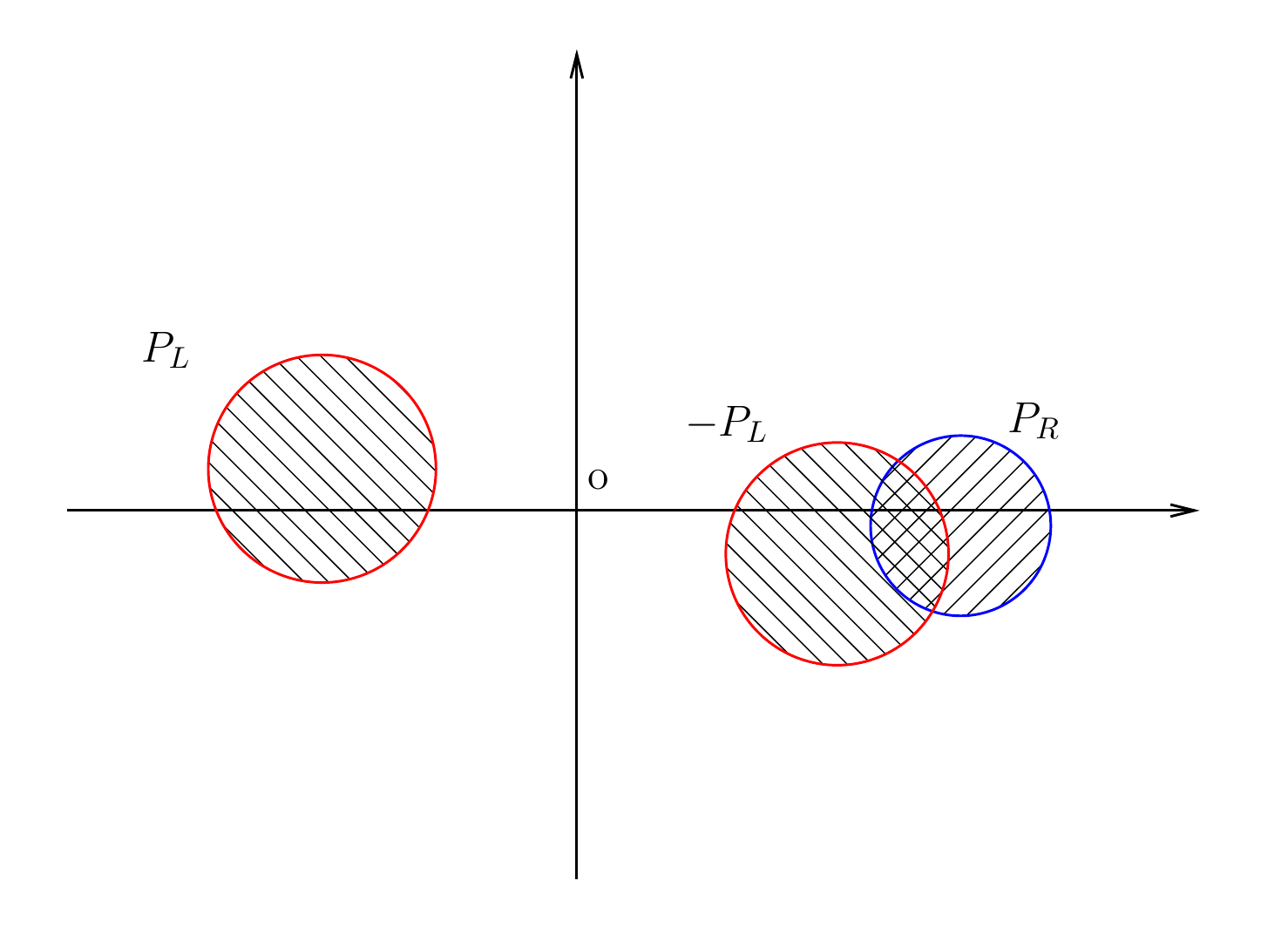}
  \vspace{-0.1in} 
    \caption{An example illustrating Theorem \ref{the-ralg1}}
  \label{fig-symmetric}
  \vspace{-0.15in}
  \end{figure}
 \vspace{-0.15in} 
  \begin{proof}
  Without loss of generality, we assume that $o$ is the origin. Since $\mathcal{F}_{opt}$ passes through  $o$, $\mathcal{F}_{opt}$ passes through the origin. 
  Thus, we can assume $\mathcal{F}_{opt}$ is the $j$-dimensional subspace spanned by the first $j$ dimensions. 
  Let $(x^i_1, \cdots, x^i_d)$ be the coordinates of each point $p_{i}\in P$,  and $\delta^2_{l}=\frac{1}{n}\sum^n_{i=1}(x^i_l)^2$ for $1\leq l\leq j$. 
   
We consider the division $P=P_{L}\cup P_{R}$ (see Fig. \ref{fig-symmetric}), where $P_{L}=\{p_{i}\in P| x^i_1< 0\}$ and $P_{R}=\{p_{i}\in P| x^i_1\geq 0\}$. Let $-P_{L}=\{-p_{i}|p_{i}\in P_{L}\}$. Consider the point set $-P_{L}\cup P_{R}$. Since $P$ is  regular  with a regular factor  $\omega$ (which is a positive constant)  and the mean of $\{x^1_l, \cdots, x^n_l\}$ is $0$ (due to the fact that  the mean point of $P$ is the origin), the coefficient of variation of $\{x^1_l, \cdots, x^n_l\}$ is no more than $\omega$, for all $1\leq l\leq j$. 

For the sample set $\mathcal{S}$,
 we define a subset $T$  of $-\mathcal{S}\cup\mathcal{S}$ as $T=(-\mathcal{S}\cap -P_L)\cup(\mathcal{S}\cap P_R)=-(\mathcal{S}\cap P_L)\cup(\mathcal{S}\cap P_R)$. Since $\mathcal{S}=(\mathcal{S}\cap P_L)\cup(\mathcal{S}\cap P_R)$, it is easy to see that $|T|=|\mathcal{S}|=m$. 
 Let $y$ be the mean point of $T$ with 
 coordinates $(y_1, \cdots, y_d)$. Since Algorithm Symmetric-Sampling enumerates the mean points of all subsets of $-\mathcal{S}\cup\mathcal{S}$,  $y$ is clearly  in $\overline{\mathcal{S}}$. If we denote the projection of $y$ on $\mathcal{F}_{opt}$ as $Proj(y)$,  then it is easy to know that $Proj(y)=(y_1, \cdots, y_j, 0, \cdots, 0)$. Let $\overrightarrow{v}$ be the unit vector $\frac{Proj(y)}{||Proj(y)||}=(\frac{y_1}{\sqrt{y^2_1+\cdots +y^2_j}}, \cdots, \frac{y_j}{\sqrt{y^2_1+\cdots +y^2_j}}, 0, \cdots, 0)$ .

Let $\delta^2_{opt}=\frac{1}{n}\sum^n_{i=1}||p_{i}, \mathcal{F}_{opt}||^2$. Below, we prove that $y$ is the desired point which induces a
$\Delta$-rotation for $\mathcal{F}_{opt}$ with $\Delta=\frac{5\sqrt{j}\omega}{1-\sqrt{\epsilon}}\delta_{opt}$.

By  Definition \ref{def-delta}, we know that in order to prove that $y$ induces a $\frac{5\sqrt{j}\omega}{1-\sqrt{\epsilon}}\delta_{opt}$-rotation, we just need to show that $\frac{||y-Proj(y)||}{||Proj(y)||}\leq \frac{5\sqrt{j}\omega}{1-\sqrt{\epsilon}}\frac{\delta_{opt}}{\sqrt{\frac{1}{n}\sum^n_{i=1}|<p_{i}, \overrightarrow{v}>|^2}}$. In other words, we need to prove two things: (a) $\frac{||Proj(y)||}{\sqrt{\frac{1}{n}\sum^n_{i=1}|<p_{i}, \overrightarrow{v}>|^2}}$ is larger than certain value, and (b) $\frac{||y-Proj(y)||}{ \delta_{opt}}$ is smaller than certain value.

 In order to prove (a), we have 
  \begin{eqnarray}
   \frac{1}{n}\sum^n_{i=1}|<p_{i}, \overrightarrow{v}>|^2=\frac{1}{n}\sum^n_{i=1}(\sum^j_{l=1}\frac{y_l x^i_l}{\sqrt{y^2_1+\cdots, y^2_j}})^2. \nonumber\\
    \leq \frac{1}{n}\sum^n_{i=1}( (\frac{\sum^j_{l=1} y^2_l}{y^2_1+\cdots +y^2_j})(\sum^j_{l=1}(x^i_l)^2))=\frac{1}{n}\sum^n_{i=1}\sum^j_{l=1}(x^i_l)^2=\sum^j_{l=1} \delta^2_{l}, \label{for-35}
   \end{eqnarray}
 where the inequality follows from the fact that $(\sum^n_{i=1}a_i b_i)^2\leq (\sum^n_{i=1}a^2_i)(\sum^n_{i=1}b^2_i)$ for any real numbers $\{a_i\mid 1\leq i\leq n\}$ and $\{b_i\mid 1\leq i\leq n\}$.
  
Meanwhile, since $||Proj(y)||=\sqrt{y^2_1+\cdots+y^2_j}$, by (\ref{for-35}) we have
  
  \begin{eqnarray}
 \frac{||Proj(y)||}{\sqrt{\frac{1}{n}\sum^n_{i=1}|<p_{i}, \overrightarrow{v}>|^2}}\geq\sqrt{\frac{y^2_1+\cdots+y^2_j}{\sum^j_{l=1} \delta^2_{l}}}. \label{for-36}
  \end{eqnarray}
  
Without loss of generality, we assume that $\delta_1=\max\{\delta_1, \cdots, \delta_j\}$. Then we have $\sqrt{\frac{y^2_1+\cdots+y^2_j}{\sum^j_{l=1} \delta^2_{l}}}\ge\sqrt{\frac{y^2_1}{\sum^j_{l=1} \delta^2_{l}}}\geq \frac{1}{\sqrt{j}}\sqrt{\frac{y^2_1}{ \delta^2_{1}}}$. Combining this  with (\ref{for-36}), we have
 \begin{eqnarray}
 \frac{||Proj(y)||}{\sqrt{\frac{1}{n}\sum^n_{i=1}|<p_{i}, \overrightarrow{v}>|^2}}\geq\frac{1}{\sqrt{j}}\frac{y_1}{ \delta_{1}}. \label{for-37}
 \end{eqnarray}
 
Since $y$ is the mean point of $T$ with $|T|=m$, 
and $T=-(\mathcal{S}\cap P_L)\cup(\mathcal{S}\cap P_R)$,  we have

\begin{eqnarray*}
y_1=\frac{\sum_{p_i\in P_L\cap \mathcal{S}}(-x^i_1)+\sum_{p_i\in P_R\cap \mathcal{S}}x^i_1}{m}=\frac{\sum_{p_i\in \mathcal{S}}|x^i_1|}{m}.
\end{eqnarray*}

Further, since $\{x^i_1\mid 1\leq i\leq n\}$ has average value zero, and its CV is bounded by $\omega$, by Lemma \ref{lem-bound}, we know that $Prob(y_1\ge  (1-\eta\sqrt{\frac{\omega^2-1}{m}})\frac{\delta_1}{\omega})\ge 1-\frac{1}{\eta^2}$. Thus, with (\ref{for-37}),  we have the following result for (a),

$$Prob(\frac{||Proj(y)||}{\sqrt{\frac{1}{n}\sum^n_{i=1}|<p_{i}, \overrightarrow{v}>|^2}}\ge \frac{1-\eta\sqrt{\frac{\omega^2-1}{m}}}{\sqrt{j}\omega})\ge 1-\frac{1}{\eta^2}.$$ 

For (b), following the same approach given in the proof of Claim (1), we can easily get a similar result as Claim (1) :
$\frac{||y-Proj(y)||}{ \delta_{opt}}\leq 5$ with probability $(1-\frac{1}{m})^2$.


Combining the above results for  (a) and (b), and setting the sample size $m=\frac{\omega^2-1}{\epsilon^2}$ and $\eta=\frac{1}{\sqrt{\epsilon}}$,  we have the following inequality, with probability $(1-\epsilon)(1-\frac{\epsilon^2}{\omega^2-1})^2$, 
$$\frac{||y-Proj(y)||}{||Proj(y)||}\leq \frac{5\delta_{opt}}{\frac{1-\sqrt{\epsilon}}{\sqrt{j}\omega}\sqrt{\frac{1}{n}\sum^n_{i=1}|<p_{i}, \overrightarrow{v}>|^2}}.$$
This means that $y$ induces a $\Delta$-rotation for $\mathcal{F}_{opt}$, where $\Delta=\frac{5\sqrt{j}\omega}{1-\sqrt{\epsilon}}\delta_{opt}$. By Lemma \ref{lem-delta}, we get the desired result, i.e., $\sum_{p\in P}||p, \mathcal{F}'||^2\leq (1+\frac{5\sqrt{j}\omega}{1-\sqrt{\epsilon}})^2\sum_{p\in P}||p, \mathcal{F}_{opt}||^2$.
 \qed
  \end{proof}
  
 By Theorem \ref{the-ralg1} and a similar idea with Theorem \ref{the-rpa}, we obtain the following theorem for regular single $j$-flat fitting problem.
\begin{theorem}
\label{the-rrpa}
Let $P$ be the 
point set of a regular  $j$-flat fitting problem in $\mathbb{R}^d$ with regular factor $\omega$. 
Then if run Algorithm Recursive-Projection on $P$ with sample size $m=\frac{\omega^2-1}{(\epsilon/j)^2}$, with probability $(1-\epsilon)(1-\frac{\epsilon^2}{\omega^2-1})^2$, the output $\mathcal{T}$ contains a root-to-leaf path such that the $j$ points associated with the path determine a flat $\mathcal{F}$ satisfying inequality
$$\sum_{p\in P}||p, \mathcal{F}||^2\leq (1+\frac{5\sqrt{j}\omega}{1-\sqrt{\epsilon}})^{2j}\sum_{p\in P}||p, \mathcal{F}_{opt}||^2.$$
\end{theorem}
\begin{proof}

Without loss of generality, we assume that $o$ is the origin of $\mathbb{R}^d$. Since all  related flats in Algorithm Recursive-Projection pass through $o$, we can view every flat as a subspace in $\mathbb{R}^d$. 

From the algorithm, we know that for any node $v$ at level $l$ of $\mathcal{T}$, $1\leq l\leq j$, there is a corresponding implicit point set $\tilde{P}_v$, which is the projection of $P$ on $f_v$. There is also an implicit flat $f_v\cap\mathcal{F}_{opt}$ in $f_v$. By Theorem \ref{the-ralg1}, we know that there is one child of $v$, denoted as $v'$, such that the rotation of $f_v\cap\mathcal{F}_{opt}$ induced by $\overrightarrow{o t}_{v'}$ forms a $\Delta$-rotation with respect to  $\tilde{P}_v$, where $\Delta=\frac{5\sqrt{j}\omega}{1-\sqrt{\epsilon}}$. 
Thus, if we always select such children (i.e., satisfying the above  condition)  from root to leaf,  we get a path with nodes $\{v_0, v_1, \cdots, v_j\}$, where $v_0$ is the root and $v_{l+1}$ is the child of $v_l$. Correspondingly, a sequence of implicit point sets $\{\tilde{P}_0, \tilde{P}_1, \cdots, \tilde{P}_j\}$ and a sequence of flats $\{\tilde{\mathcal{F}}_0, \tilde{\mathcal{F}}_1, \cdots, \tilde{\mathcal{F}}_j\}$ also be  obtained, which have the following properties. 

\begin{enumerate}
\item Initially, $\tilde{\mathcal{F}}_0=\mathcal{F}_{opt}$, and $\tilde{P}_0=P$.

\item For any $1\leq l\leq j$, $\tilde{\mathcal{F}}_l$ is the $\Delta$-rotation of $\tilde{\mathcal{F}}_{l-1}\cap f_{v_{l-1}}$ induced by $\overrightarrow{o t}_{v_{l}}$, and $\tilde{P}_l$ is the projection of $\tilde{P}_{l-1}$ on $f_{v_{l-1}}$ (see Figure \ref{fig-project})
 \end{enumerate}

Note that since both  $\tilde{\mathcal{F}}_l$ and $\tilde{\mathcal{F}}_{l-1}\cap f_{v_{l-1}}$ locate on $f_{v_{l-1}}$, they are all perpendicular to $\overrightarrow{o t}_{v_{l-1}}$.


The following claim reveals 
the dimensionality of each $\tilde{\mathcal{F}}_l$.

\begin{claim}[\textbf{5}]
For $1\leq l\leq j$, $\tilde{\mathcal{F}}_l$ is a $(j-l+1)$-dimensional subspace.
\end{claim}

\begin{proof}
We prove the claim by induction. For the base case (i.e., $l=1$), $\tilde{\mathcal{F}}_1$ has the same dimensionality as $\tilde{\mathcal{F}}_0\cap f_{v_0}=\mathcal{F}_{opt}$ (note $f_{v_0}=\mathbb{R}^d$). Hence, $\tilde{\mathcal{F}}_1$ is a $j$-dimensional subspace. Then we assume that $\tilde{\mathcal{F}}_w$ is a $(j-w+1)$-dimensional subspace for $w\leq l$. Now we consider the case of $l+1$. Since $\tilde{\mathcal{F}}_{l+1}$ is a rotation of $\tilde{\mathcal{F}}_{l}\cap f_{v_{l}}$,  they have the same dimensionality. Also, from the algorithm, we know that $\tilde{\mathcal{F}}_{l}\cap f_{v_{l}}$ is the subspace in $\tilde{\mathcal{F}}_{l}$ which is perpendicular to $Proj(\overrightarrow{o t}_{v_{l}})$, where $Proj(\overrightarrow{ot}_{v_{l}})$ is the projection of $\overrightarrow{ot}_{v_{l}}$ on $\tilde{\mathcal{F}}_{l}$. Thus, $\tilde{\mathcal{F}}_{l}\cap f_{v_{l}}$ is a $(j-l+1-1)$-dimensional subspace, which implies that $\tilde{\mathcal{F}}_{l+1}$ is also a $(j-l)$-dimensional subspace. Hence, Claim (\textbf{5}) is proved.
\qed
\end{proof}

By Lemma \ref{lem-delta}, we have the following claim.
\begin{claim}[\textbf{6}]
For any $1\leq l\leq j$, $\sum_{\tilde{p}\in\tilde{P}_l}||\tilde{p},\tilde{ \mathcal{F}}_l||^2\leq (1+\frac{5\sqrt{j}\omega}{1-\sqrt{\epsilon}})^2\sum_{\tilde{p}\in\tilde{P}_l}||\tilde{p}, \tilde{\mathcal{F}}_{l-1}\cap f_{v_{l-1}}||^2$.
\end{claim}

We construct another sequence of flats $\{\mathcal{F}_1,\cdots, \mathcal{F}_j \}$ as follows.

\begin{enumerate}
\item Initially, $\mathcal{F}_1=\tilde{\mathcal{F}}_1$.
\item For any $2\leq l\leq j$, $\mathcal{F}_l=span\{\tilde{\mathcal{F}}_l, \overrightarrow{ot}_{v_1}, \cdots, \overrightarrow{o t}_{v_{l-1}}\}$.
\end{enumerate}

From the above construction for $\mathcal{F}_l$, we have the following claim.

\begin{claim}[\textbf{7}]
For any point $\tilde{p}\in \tilde{P}_l$, let $p$ be the corresponding point when $\tilde{p}$ is mapped back to $\mathbb{R}^d$. Then, $||p, \mathcal{F}_l||=||\tilde{p}, \tilde{\mathcal{F}}_l||$. 
\end{claim}

From Claim (\textbf{5}), we know that each $\mathcal{F}_l$ is a $j$-dimensional subspace. From the algorithm, we know that $span\{\tilde{\mathcal{F}}_{l-1}\cap f_{v_{l-1}}, \overrightarrow{o t}_{v_{l-1}}\}=\tilde{\mathcal{F}}_{l-1}$ (see Fig. \ref{fig-project}), which implies $span\{\tilde{\mathcal{F}}_{l-1}\cap f_{v_{l-1}}, \overrightarrow{ot}_{v_1}, \cdots,$ $\overrightarrow{o t}_{v_{l-1}}\}=\mathcal{F}_{l-1}$. Then by Claim (\textbf{6}) and (\textbf{7}), we have the following inequality.
\begin{eqnarray}
\sum_{p\in P}||p, \mathcal{F}_l||^2\leq (1+\frac{5\sqrt{j}\omega}{1-\sqrt{\epsilon}})^2\sum_{p\in P}||p, \mathcal{F}_{l-1}||^2. \label{for-39}
\end{eqnarray}

%
%
%

 Recursively using inequality (\ref{for-39}), we have $\sum_{p\in P}||p, \mathcal{F}_j||^2\leq (1+\frac{5\sqrt{j}\omega}{1-\sqrt{\epsilon}})^{2(j-1)}\sum_{p\in P}||p, \mathcal{F}_1||^2$. Meanwhile, we have $\sum_{p\in P}||p, \mathcal{F}_1||^2\leq (1+\frac{5\sqrt{j}\omega}{1-\sqrt{\epsilon}})^2\sum_{p\in P}||p, \mathcal{F}_{opt}||^2$. Thus, 
 $$\sum_{p\in P}||p, \mathcal{F}_j||^2\leq (1+\frac{5\sqrt{j}\omega}{1-\sqrt{\epsilon}})^{2j}\sum_{p\in P}||p, \mathcal{F}_{opt}||^2.$$
 
By Claim (\textbf{5}), we know that $\tilde{\mathcal{F}}_j$ is a $1$-dimension flat (i.e., the single line spanned by $\overrightarrow{o t}_{v_{j}}$). Hence, $\mathcal{F}_j=span\{\overrightarrow{ot}_{v_1}, \cdots, \overrightarrow{o t}_{v_{j}}\}$. Thus if setting $\mathcal{F}=\mathcal{F}_j$,  we have 
$$\sum_{p\in P}||p, \mathcal{F}||^2\leq (1+\frac{5\sqrt{j}\omega}{1-\sqrt{\epsilon}})^{2j}\sum_{p\in P}||p, \mathcal{F}_{opt}||^2.$$
 
\noindent\textbf{Success probability:} Each time we use Theorem \ref{the-ralg1}, the success probability is $(1-\frac{\epsilon}{j})(1-\frac{(\epsilon/j)^2}{\omega^2-1})^2$ (we replace $\epsilon$ by $\frac{\epsilon}{j}$ tp increase the sample size from $\frac{\omega^2-1}{\epsilon^2}$ to $\frac{\omega^2-1}{(\epsilon/j)^2}$). Thus, the total success probability is $((1-\frac{\epsilon}{j})(1-\frac{(\epsilon/j)^2}{\omega^2-1})^2)^{j}\geq (1-\epsilon)(1-\frac{\epsilon^2}{\omega^2-1})^2$.
\qed
\end{proof}

With the above theorem, we can easily have the following theorem (using the approach similar to Theorem \ref{the-ptas}).   
\begin{theorem}
\label{the-rptas}
Let $P$ be the point set of a regular $(k,j)$-projective clustering problem in $\mathbb{R}^d$ with regular factor $\omega$. If each optimal cluster has at least   $\alpha|P|$ points from $P$ for some constant  $0<\alpha \le  \frac{1}{k}$, Algorithm Projective-Clustering  yields a PTAS with constant probability, where the running time of the PTAS is $O(2^{poly(\frac{kj\omega}{\epsilon\alpha})}nd)$.

\end{theorem}
\section{Extension to $L_\tau$ Sense Projective Clustering}

We first introduce the $L_\tau$ sense $\Delta$-rotation.

\begin{definition}[$L_\tau$ Sense $\Delta$-Rotation]
\label{def-ldelta}
Let $P$ be a points set, $\mathcal{F}$ be a $j$-dimensional flat, and $u-o$ be a vector in $\mathbb{R}^d$ with $o\in\mathcal{F}$. Further, let $h^\tau=\frac{1}{|P|}\sum_{p\in P}|<p-o, \frac{Proj(u)-o}{||Proj(u)-o||}>|^\tau$, and $\mathcal{F}'$ be a rotation of $\mathcal{F}$ induced by $u$ with angle $\theta$,  where $Proj(u)$ is the projection of $u$ on $\mathcal{F}$.
Then, $\mathcal{F}'$ is a $\Delta$-rotation of $\mathcal{F}$ (with respect to $P$) if $\theta \leq \arctan\frac{\Delta}{h^{\tau}}$.
\end{definition}

The following lemma shows t how the value of $\frac{1}{|P|}\sum_{p\in P}||p, \mathcal{F}||^\tau$ changes after a $L_{\tau}$ sense $\Delta$-rotation.

\begin{lemma}
\label{lem-ldelta}
Let $P$ be a point set, $\mathcal{F}$ be a $j$-dimensional flat, and $u$ be a point in $\mathbb{R}^d$. If $\mathcal{F}'$ is a $L_{\tau}$ sense $\Delta$-rotation (with respect to $P$) of $\mathcal{F}$ induced by the vector $u-o$ for some point $o\in\mathcal{F}$,  then for any integer $1\leq \tau<\infty$, 
$$(\frac{1}{|P|}\sum_{p\in P}||p, \mathcal{F}'||^\tau)^{1/\tau} \leq(\frac{1}{|P|}\sum_{p\in P}||p, \mathcal{F}||^\tau)^{1/\tau}+\Delta.$$
\end{lemma}

Before proving Lemma \ref{lem-ldelta}, we first  introduce the following lemma.

\begin{lemma}
\label{lem-inequal}
For any integer $\tau\ge 1$, and positive numbers $x,y, \alpha$, 
 $(x+\alpha y)^\tau\leq (1+\alpha)^{\tau-1}x^\tau+\alpha(1+\alpha)^{\tau-1}y^{\tau}.$
\end{lemma}
\begin{proof}
We prove this lemma by mathematical induction on $\tau$.\\

\noindent\textbf{Base case:} For $\tau=1$, it is easy to see that $(x+\alpha y)^\tau=(1+\alpha)^{0}x^1+\alpha(1+\alpha)^{0}y^{1}=(1+\alpha)^{\tau-1}x^\tau+\alpha(1+\alpha)^{\tau-1}y^{\tau}$. Thus, base case holds.\\

\noindent\textbf{Induction step:} Assume that the inequality holds for $\tau\le \tau_0$ for some $\tau_{0} \ge 1$ (i.e., Induction hypothesis). Now consider the case of $\tau=\tau_0+1$. By the induction hypothesis, we have
\begin{eqnarray}
(x+\alpha y)^{\tau_0+1}\leq (x+\alpha y)((1+\alpha)^{\tau_0-1}x^{\tau_0}+\alpha(1+\alpha)^{\tau_0-1}y^{\tau_0}) \nonumber\\
=(1+\alpha)^{\tau_0-1}x^{\tau_0+1}+\alpha(1+\alpha)^{\tau_0-1}(x^{\tau_0}y+xy^{\tau_0})+\alpha^2(1+\alpha)^{\tau_0-1}y^{\tau_0+1}. \label{for-47}
\end{eqnarray}

Since both $x$ and $y$ are positive, we have $(x-y)(x^{\tau_0}-y^{\tau_0})\geq 0$.
Also, it is easy to know that
$$(x-y)(x^{\tau_0}-y^{\tau_0})\geq 0\Longleftrightarrow x^{\tau_0+1}+y^{\tau_0+1}\geq x^{\tau_0}y+xy^{\tau_0}.$$
Thus, if replacing $x^{\tau_0}y+xy^{\tau_0}$ by $x^{\tau_0+1}+y^{\tau_0+1}$ in (\ref{for-47}), we have
$$(x+\alpha y)^{\tau_0+1}\leq(1+\alpha)^{\tau_0-1}x^{\tau_0+1}+\alpha(1+\alpha)^{\tau_0-1}(x^{\tau_0+1}+y^{\tau_0+1})+\alpha^2(1+\alpha)^{\tau_0-1}y^{\tau_0+1}$$
$$=(1+\alpha)^{\tau_0}x^{\tau_0+1}+\alpha(1+\alpha)^{\tau_0}y^{\tau_0+1}.$$
Hence, the inequality holds for $\tau=\tau_0+1$.
\qed
\end{proof}

Now we prove Lemma \ref{lem-ldelta}.
\begin{proof}[\textbf{of Lemma \ref{lem-ldelta}}]
We use the same notations as in Definition \ref{def-ldelta}. For any $p\in P$, let $Proj( p)$ be its projection on $\mathcal{F}$ and $u_{p}=|<p-o, \frac{Proj(u)-o}{||Proj(u)-o||}>|$. Then we have

$$||p, \mathcal{F}'||\leq ||p-Proj(p)||+||Proj(p), \mathcal{F}'||=||p, \mathcal{F}||+||Proj(p), \mathcal{F}'||.$$

Since the rotation angle from $\mathcal{F}$ to $\mathcal{F}'$ is $\theta\leq\arctan\frac{\Delta}{h^{\tau}}$, we have $||Proj(p), \mathcal{F}'||=u_p \sin\theta\leq u_p\tan \theta\leq  \frac{\Delta}{h^{\tau}}u_p$. Let $\delta=(\frac{1}{|P|}\sum_{p\in P}||p, \mathcal{F}||^\tau)^{1/\tau}$. Then, we have 
$$||p, \mathcal{F}'||^\tau\leq (||p, \mathcal{F}||+u_p \sin\theta)^\tau\leq(||p, \mathcal{F}||+\frac{\Delta}{\delta}\frac{\delta}{h}u_p)^\tau.$$
Using Lemma \ref{lem-inequal} with $x=||p, \mathcal{F}||$, $y=\frac{\delta}{h}u_p$, and $\alpha=\frac{\Delta}{\delta}$, we have 
\begin{eqnarray}
||p, \mathcal{F}'||^\tau\leq(1+\frac{\Delta}{\delta})^{\tau-1}||p, \mathcal{F}||^\tau+\frac{\Delta}{\delta}(1+\frac{\Delta}{\delta})^{\tau-1}(\frac{\delta}{h}u_p)^{\tau}. \label{for-50}
\end{eqnarray}
Summing both sides of  (\ref{for-50}) over $p$,  we have 
\begin{eqnarray*}
\sum_{p\in P}||p, \mathcal{F}'||^\tau &\leq& (1+\frac{\Delta}{\delta})^{\tau-1}\sum_{p\in P}||p, \mathcal{F}||^\tau+\frac{\Delta}{\delta}(1+\frac{\Delta}{\delta})^{\tau-1}\sum_{p\in P}(\frac{\delta}{h}u_p)^\tau\\
&=&(1+\frac{\Delta}{\delta})^{\tau-1}\sum_{p\in P}||p, \mathcal{F}||^\tau+\frac{\Delta}{\delta}(1+\frac{\Delta}{\delta})^{\tau-1}(\frac{\delta}{h})^\tau\sum_{p\in P}(u_p)^\tau.
\end{eqnarray*}
Since $h^\tau=\frac{1}{|P|}\sum_{p\in P}(u_p)^\tau$ and $\delta^\tau=\frac{1}{|P|}\sum_{p\in P}||p, \mathcal{F}||^\tau$,  the above inequality becomes
\begin{eqnarray*}
\frac{1}{|P|}\sum_{p\in P}||p, \mathcal{F}'||^\tau&\leq&(1+\frac{\Delta}{\delta})^{\tau-1}\frac{1}{|P|}\sum_{p\in P}||p, \mathcal{F}||^\tau+\frac{\Delta}{\delta}(1+\frac{\Delta}{\delta})^{\tau-1}\frac{1}{|P|}\sum_{p\in P}||p, \mathcal{F}||^\tau\\
&=&(1+\frac{\Delta}{\delta})^\tau\frac{1}{|P|}\sum_{p\in P}||p, \mathcal{F}||^\tau=(\delta+\Delta)^\tau.
\end{eqnarray*}
Thus the lemma is true.
\qed
\end{proof}

Using Lemma \ref{lem-ldelta} and a similar approach for the $L_2$ case, we have the following theorem. Since the idea and proofs are almost the same, we omit them from the paper.
\begin{theorem}
\label{the-lptas}
Let $P$ be the point set of an $L_\tau$ sense $(k,j)$-projective clustering problem in $\mathbb{R}^d$ for integer $1\leq \tau<\infty$. Let $Opt$ be the optimal objective value. With constant probability and in $O(2^{poly(\frac{kj}{\epsilon\gamma})}nd)$ time,  Algorithm Projective-Clustering outputs an approximation solution $\{\mathcal{F}_1, \cdots,$ $\mathcal{F}_k\}$ such that each $\mathcal{F}_l$ is a $j$-flat, and $\frac{1}{|P'|}\sum_{p\in P'}\min_{1\leq l\leq k}||p, \mathcal{F}_l||^\tau\leq (1+\epsilon) Opt$,
where $P'$ is a subset of $P$ with   at least $(1-\gamma)|P|$ points. 
\end{theorem}

Furthermore, we also have a similar result for $L_\tau$ sense regular projective clustering.  For the same reason, we  omit the details for this case.

\begin{definition}[ $L_\tau$ Sense Coefficient of Variation (CV) ]
\label{def-lcv}
Let $x$ be a random variable, and $\mu=E[x]$. The coefficient of variation of $x$ is denoted as $\frac{(E[|x-\mu|^\tau])^{\frac{1}{\tau}}}{E[|x-\mu|]}$.
\end{definition}

\begin{lemma}
\label{lem-lbound}
Let $\mathcal{X}=\{x_{1}, \cdots, x_{n}\}$ be a set of $n$ numbers with coefficient of variation $\omega$, and $S=\{x_{i_{1}}, \cdots, x_{i_{m}}\}$ be a random sample of $\mathcal{X}$. 
Also, let $\mu=\frac{1}{n}\sum^n_{i=1}x_i$ and $h^\tau=\frac{1}{n}\sum^n_{i=1}(x_i-\mu)^\tau$. Then, for any
positive constant $\eta$ and integer $1\leq \tau<\infty$, 
$$Prob(\frac{\sum^{m}_{l=1}|x_{i_l}-\mu|}{m}\ge (1-\eta\sqrt{\frac{\omega^2-1}{m}})\frac{h}{\omega})\ge 1-\frac{1}{\eta^2}.$$

\end{lemma}
\begin{theorem}
\label{the-lrptas}
Let $P$ be the point set of an $L_\tau$ sense regular $(k,j)$-projective clustering problem in $\mathbb{R}^d$ with 
regular factor $\omega$, where  $1\leq \tau<\infty$ is an integer. If each optimal cluster has size at least $\alpha|P|$. Then  Algorithm Projective-Clustering yields a PTAS with constant probability, where the running time of the PTAS is $O(2^{poly(\frac{kj\omega}{\epsilon\alpha})}nd)$.

\end{theorem}

\newpage
\begin{thebibliography}{7}


\bibitem{AHV04}
Pankaj K. Agarwal, Sariel Har-Peled, Kasturi R. Varadarajan, '' Approximating extent measures of points.'' {\em J. ACM 51(4):606-635}, 2004

\bibitem{AHV05}
P.K. Agarwal, S. Har-Peled, and K. R. Varadarajan, ``Geometric
Approximation via Coresets",
Combinatorial and Computational Geometry, MSRI Publications Volume 52, pp.~1--30, 2005.


\bibitem{AMV}
Pankaj K. Agarwal, Cecilia Magdalena Procopiuc, Kasturi R. Varadarajan, '' Approximation Algorithms for a k-Line Center''. {\em Algorithmica 42(3-4): 221-230 (2005)}

\bibitem{AM04}
Pankaj K. Agarwal and Nabil H. Mustafa, '' k-meansprojective
clustering''. {\em In Proceedings of the twenty-third ACM
SIGMOD-SIGACT-SIGART symposiumon Principles of database systems, PODS
'04, pages155-165, New York, NY, USA, 2004. ACM.}

\bibitem{AS92}
N. Alon and J. H. Spencer, '' The Probabilistic Method. John Wiley and Sons'', 1992.


\bibitem{AWY}
Charu C. Aggarwal, Joel L. Wolf, Philip S. Yu, Cecilia
Procopiuc, and Jong Soo Park. ''Fast algorithms for
projected clustering''. {\em In Proceedings of the 1999 ACM
SIGMOD international conference on Management of
data, SIGMOD '99, pages 61-72, New York, NY, USA,
1999. ACM.}

\bibitem{AY00}
Charu C. Aggarwal and Philip S. Yu. '' Finding gener-
alized projected clusters in high dimensional spaces''.
{\em In Proceedings of the 2000 ACM SIGMOD interna-
tional conference on Management of data, SIGMOD
'00, pages 70-81, New York, NY, USA, 2000. ACM.}




\bibitem{BC03}
M.Badoiu, K.Clarkson, ``Smaller core-sets for balls", {\em Proceedings
of the fourteenth annual ACM-SIAM symposium on Discrete algorithms}, pp.~801--802,
2003.

\bibitem{BHI02}
M.Badoiu, S.Har-Peled, P.Indyk, ``Approximate clustering via
core-sets", {\em Proceedings of the 34th Symposium on Theory of
Computing}, pp.~250--257,
2002.

\bibitem{Cl08}
K. Clarkson, ``Coresets, sparse greedy approximation, and the Frank-Wolfe algorithm", {\em Proceedings
of the nineteenth annual ACM-SIAM symposium on Discrete algorithms}, pp.~922-931, 2008

\bibitem{DV07}
Amit Deshpande, Kasturi R. Varadarajan, '' Sampling-based dimension reduction for subspace approximation''. {\em STOC 2007: 641-650}

\bibitem{DRV}
Amit Deshpande, Luis Rademacher, Santosh Vempala, Grant Wang, ''Matrix approximation and projective clustering via volume sampling''. {\em SODA 2006: 1117-1126}

\bibitem{EV05}
Michael Edwards, Kasturi R. Varadarajan, ''No Coreset, No Cry: II''. {\em FSTTCS 2005: 107-115}

\bibitem{FFS}
Dan Feldman, Amos Fiat, Micha Sharir, Danny Segev: Bi-criteria linear-time approximations for generalized k-mean/median/center. Symposium on Computational Geometry 2007: 19-26

\bibitem{FKV}
Alan M. Frieze, Ravi Kannan, Santosh Vempala: Fast monte-carlo algorithms for finding low-rank approximations. J. ACM 51(6): 1025-1041 (2004)

\bibitem{FL11}
Dan Feldman, Michael Langberg: A unified framework for approximating and clustering data. STOC 2011: 569-578

\bibitem{HRZ}
Sariel Har-Peled, Dan Roth, Dav Zimak, '' Maximum Margin Coresets for Active and Noise Tolerant Learning.'' {\em  IJCAI 2007: 836-841}


\bibitem{KMY02}
P.Kumar, J.Mitchell, A.Yildirim, ``Computing Core-Sets and Approximate
Smallest Enclosing Hyperspheres in High Dimensions", manuscript, 2002.

\bibitem{HW04}
Sariel Har-Peled, Yusu Wang, '' Shape Fitting with Outliers''. {\em SIAM J. Comput. (SIAMCOMP) 33(2):269-285 (2004)}

\bibitem{HV03}
Sariel Har-Peled, Kasturi R. Varadarajan, '' High-dimensional shape fitting in linear time''. {\em SoCG 2003:39-47}

\bibitem{HV02}
Sariel Har-Peled, Kasturi R. Varadarajan, '' Projective clustering in high dimensions using core-sets''. {\em Symposium on Computational Geometry 2002: 312-318}

\bibitem{IKI}
Mary Inaba, Naoki Katoh, Hiroshi Imai, '' Applications of Weighted Voronoi Diagrams and Randomization to Variance-Based k-Clustering (Extended Abstract)''. {\em Symposium on Computational Geometry 1994: 332-339}



\bibitem{KSS04}
Amit Kumar, Yogish Sabharwal, Sandeep Sen: A Simple Linear Time $(1+\epsilon)$-Approximation Algorithm for k-Means Clustering in Any Dimensions. FOCS 2004: 454-462


\bibitem{KSS05}
Amit Kumar, Yogish Sabharwal, Sandeep Sen: Linear Time Algorithms for Clustering Problems in Any Dimensions. ICALP 2005: 1374-1385





\bibitem{KSS10}
Amit Kumar, Yogish Sabharwal, Sandeep Sen: Linear-time approximation schemes for clustering problems in any dimensions. J. ACM 57(2): (2010)



\bibitem{PJA}
Cecilia M. Procopiuc, Michael Jones, Pankaj K. Agar-
wal, and T. M. Murali. ''A monte carlo algorithm for
fast projective clustering''. {\em In Proceedings of the 2002
ACM SIGMOD international conference on Manage-
ment of data, SIGMOD '02, pages 418-427, New York,
NY, USA, 2002. ACM.}


\bibitem{SV07}
Nariankadu D. Shyamalkumar, Kasturi R. Varadarajan, '' Efficient subspace approximation algorithms''. {\em SODA 2007: 532-540}


\bibitem{VX12}
Kasturi Varadarajan, Xin Xiao, ''A near-linear algorihtms for projective clustering integer points'', {\em SODA 2012}
\end{thebibliography}
\end{document}